\documentclass{article}

\usepackage{arxiv}

\usepackage[utf8]{inputenc} 
\usepackage[T1]{fontenc}    
\usepackage{hyperref}       
\usepackage{url}            
\usepackage{booktabs}       
\usepackage{amsfonts}       
\usepackage{nicefrac}       
\usepackage{microtype}      
\usepackage{graphicx}
\usepackage{doi}

\usepackage{times}
\usepackage{color}
\usepackage{multirow}
\usepackage[dvipsnames]{xcolor}
\usepackage[authoryear]{natbib}
\usepackage{rotating}
\usepackage{bbm}
\usepackage{latexsym}

\usepackage{hyperref}
\hypersetup{colorlinks=true, linkcolor=gnblue4, breaklinks=true, urlcolor=gnblue4, citecolor=gnblue4}

\newcommand{\subscr}[2]{{#1}_{\textup{#2}}}

\newcommand{\until}[1]{\{1,\dots,#1\}}

\usepackage{bbm}
\newcommand{\de}{\mathrm{d}}

\newcommand{\real}{\mathbb{R}}

\renewcommand{\real}{\mathbb{R}}

\newcommand{\I}{\mathrm{I}}

\newcommand{\diag}{\mathrm{diag}}

\newcommand{\volt}{\mathrm{v}}

\DeclareSymbolFont{bbold}{U}{bbold}{m}{n}
\DeclareSymbolFontAlphabet{\mathbbold}{bbold}

\newcommand{\vect}[1]{\mathbbold{#1}}
\newcommand{\vectorones}[1][]{\vect{1}_{#1}}
\newcommand{\vectorzeros}[1][]{\vect{0}_{#1}}

\definecolor{gnblue1}{RGB}{0,36,71}   
\definecolor{gnblue2}{RGB}{0,60,118}  
\definecolor{gnblue3}{RGB}{0,85,164}  
\definecolor{gnblue4}{RGB}{0,108,212} 
\definecolor{gnblue5}{RGB}{0,133,255}  
\definecolor{gnblue6}{RGB}{35,156,255} 
\definecolor{gnblue7}{RGB}{88,177,255} 
\definecolor{gnbrown1}{RGB}{71,27,0}  
\definecolor{gnbrown2}{RGB}{117,45,0} 
\definecolor{gnbrown3}{RGB}{164,62,0} 
\definecolor{gnbrown4}{RGB}{211,80,0} 
\definecolor{gnbrown5}{RGB}{255,97,0} 
\definecolor{gnbrown6}{RGB}{255,127,26} 
\definecolor{gnbrown7}{RGB}{255,155,86} 

\usepackage[prependcaption,colorinlistoftodos]{todonotes}

\usepackage{amsmath}
\usepackage{amssymb}
\usepackage{amsthm}
\usepackage{xcolor}
\usepackage{tikz}
\usepackage{pgfplots}
\usepackage{subfig}
\usepackage{cleveref}       

\newtheorem{theorem}{Theorem}
\newtheorem{lemma}{Lemma}

\newtheorem{proposition}{Proposition}

\newtheorem{assumption}{Assumption}

\newtheorem{remark}{Remark}

\newcommand{\captionfonts}{\normalsize}

\makeatletter  
\long\def\@makecaption#1#2{%
  \vskip\abovecaptionskip
  \sbox\@tempboxa{{\captionfonts #1: #2}}%
  \ifdim \wd\@tempboxa >\hsize
    {\captionfonts #1: #2\par}
  \else
    \hbox to\hsize{\hfil\box\@tempboxa\hfil}%
  \fi
  \vskip\belowcaptionskip}
\makeatother   

\title{Firing Rate Models as Associative Memory: \newline Excitatory-Inhibitory Balance for Robust Retrieval}

\author{ \href{https://orcid.org/0009-0000-3444-0838}{\includegraphics[scale=0.06]{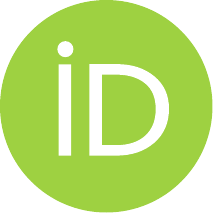}\hspace{1mm}Simone Betteti} \\
	Department of Information Engineering\\
	Università degli Studi di Padova\\
	Padova, 35131, IT \\
	\texttt{bettetisim[at]dei.unipd.it} \\
	\And
	\href{https://orcid.org/0000-0002-9439-296X}{\includegraphics[scale=0.06]{orcid.pdf}\hspace{1mm}Giacomo Baggio} \\
	Department of Information Engineering\\
	Università degli Studi di Padova\\
	Padova, 35131, IT \\
	\texttt{baggio[at]dei.unipd.it} \\
	\AND
    \href{https://orcid.org/0000-0002-4785-2118}{\includegraphics[scale=0.06]{orcid.pdf}\hspace{1mm}Francesco Bullo} \\
	Center for Control, Dynamical Systems and Computation\\
	University of California at Santa Barbara,\\
	Santa Barbara, CA, 93106 IT \\
	\texttt{bullo[at]ucsb.edu}
    \And
    \href{https://orcid.org/0000-0001-8926-1888}{\includegraphics[scale=0.06]{orcid.pdf}\hspace{1mm}Sandro Zampieri} \\
	Department of Information Engineering\\
	Università degli Studi di Padova\\
	Padova, 35131, IT \\
	\texttt{zampi[at]dei.unipd.it}\\ 
}

\renewcommand{\headeright}{}
\renewcommand{\undertitle}{}
\renewcommand{\shorttitle}{Firing Rate Models as Associative Memory: Excitatory-Inhibitory Balance for Robust Retrieval}

\hypersetup{
pdftitle={Firing Rate Models as Associative Memory: Excitatory-Inhibitory Balance for Robust Retrieval},
pdfsubject={math.DS, cs.AI, cs.ML, phys.ST, q-bio.NC},
pdfauthor={Simone Betteti, Giacomo Baggio, Francesco Bullo, Sandro Zampieri},
pdfkeywords={Associative Memory, Recurrent Neural Networks},
}

\begin{document}
\maketitle

\begin{abstract}
Firing rate models are dynamical systems widely used in applied and theoretical neuroscience to describe local cortical dynamics in neuronal populations. By providing a macroscopic perspective of neuronal activity, these models are essential for investigating oscillatory phenomena, chaotic behavior, and associative memory processes. Despite their widespread use, the application of firing rate models to associative memory networks has received limited mathematical exploration, and most existing studies are focused on specific models.  Conversely, well-established associative memory designs, such as Hopfield networks, lack key biologically-relevant features intrinsic to firing rate models, including positivity and interpretable synaptic matrices that reflect excitatory and inhibitory interactions. To address this gap, we propose a general framework that ensures the emergence of re-scaled memory patterns as stable equilibria in the firing rate dynamics. Furthermore, we analyze the conditions under which the memories are locally and globally asymptotically stable, providing insights into constructing biologically-plausible and robust systems for associative memory retrieval.
\end{abstract}


\section{Introduction}
The modelling of associative memory processes began in the early 1970's and
1980's with the mathematical formalization of Amari \citep{A:72,A:77} and
of Grossberg \citep{G:83} and the elegant and explicit construction of
Hopfield \citep{H:82, H:84}. The authors drew inspiration from the early
successes of statistical physics in the description of glassy phenomena
\citep{SK:75, MP:87}, leveraging the average properties of simple
interconnected units. The key idea was to define a network of neuron-like
computational units, similar to those studied by McCulloch and Pitts
\citep{MW:43}, and investigate memory retrieval as emergent
processes. Within this context, the authors conceptualized associative
memory networks as dynamical systems defined by ordinary differential
equations (ODEs) having as stable equilibrium points the memory patterns to
retrieve. The first key contribution was the definition of a Lyapunov
function for the associative memory system that ensured global asymptotic
convergence to the equilibria of the system. Thus, any initial condition
for the system would lie in the basin of attraction of one of these
equilibria, and the system trajectory will inevitably evolve towards
it. The second key contribution was the explicit definition of the set of
memory vectors as binary patterns taking values in $\{-1,+1\}$, analogously
to ferromagnets in spin glasses. The binary representation of the memory
vectors allowed for the explicit design of a synaptic matrix, that ensured
the system's equilibria precisely matched the intended memories. The
effective combination of the two key contributions has catalyzed a wealth
of subsequent research, both analytical and numerical, focusing on the
fundamental properties of associative memory systems. Notably, many authors
have studied the storage capacity \citep{AHM:87, AHMb:87, MPRV:87,P:95} of
associative memory systems, that is the maximum number of memories that can
be stored in the synaptic matrix without compromising their stability.
Subsequent works \citep{TF:88, T:90, AT:91} have
extended beyond the binary spin structure proposed by Hopfield, enabling
binary positive activations $\{0,1\}$ and low levels of neural
activity. These works retain the dynamic framework initially proposed by
Grossberg and Hopfield, hereafter referred to as \emph{voltage
equations}, but apply non-negative activation functions to yield neuron
firing rates.  While positive activations allow interpreting synaptic
matrix components as excitatory or inhibitory, the relationship between
voltages and firing rates remains arbitrary and highly dependent on
network parameters.  This limitation has led associative memory
researchers to explore models that directly connect empirically
measurable quantities, like firing rates, with one another. 

Parallel to the advancements in the modeling of
associative memory networks, detailed biophysical models of cortical
circuits have received an increasing amount of attention due
to their capability of generating synthetic data of EEG recordings, thus
providing insight to experimentalists on the roots of the measured
quantities. Beginning with the groundbreaking description of the neuron
biochemical response by Hodgkin and Huxley \citep{HH:52}, the
characterization of neuronal properties by means of dynamical systems has
become ever more pervasive. In the beginning, the mathematical
characterization of neurons and neural processes focused on the extensive
treatment of the microscopic properties, such as gating and diffusion of
ions, and branched into the well known FitzHugh-Nagumo \citep{FH:61} and
Morris-Lecar \citep{ML:81} models. However, the mathematical complexity of
these models, combined with the computational limitations, constrained
researchers to small-scale studies involving only a handful of neurons. To
address these challenges, substantial efforts were devoted to the
derivation of simplified mathematical models amenable to analytical
treatment and large scale simulation. This effort led to the establishment
of the class of models widely known as Integrate-and-Fire \citep{BC:00,
  B:00, B:06, BAN:06}. Despite their utility, using Integrate-and-Fire
models for memory retrieval remains challenging due to the hybrid nature of
their dynamics, switching between continuous dynamics and a hard reset to a
given initial condition. The key idea for the formulation of biologically
plausible associative memory systems was to consider an Integrate-and-Fire
model and to average neural spikes over fixed time windows to derive a rate
of firing \citep{ET:10, WG:14} for the neuron. The core
of the newly proposed biologically plausible model, referred to as the
\emph{firing rate model}, lies in its use of a non-negative activation
function that directly processes firing rates rather than membrane
voltages. This approach effectively links two empirical observables
across clusters of neurons. Despite its potential importance for
neuroscience, designing firing rate systems—such as synaptic matrices and
activation functions—so that specific memories appear as locally stable
equilibria remains an underexplored area \citep[Section~7.4]{AD:05}. 

The primary contributions of this paper are (i) the
design of a synaptic matrix that encodes memories as equilibrium points
within the firing rate system, applicable to arbitrary activation
functions, and (ii) an analysis of both local and global stability,
building upon and expanding the foundational work presented in
\citep[Section~7.4]{AD:05}. Specifically, we present a method to design
a synaptic matrix of the \emph{firing rate} model that guarantees the
retrieval of a rescaled version of the prototypical memories. These
prototypical memories are assumed to be equally sparse and equally
correlated binary vectors, meaning that they share a common number of ones
and of overlapping entries. In particular, our first theorem states the
necessary and sufficient conditions that ensure the existence of the
rescaled prototypical memories as equilibrium points for the \emph{firing
rate} dynamics. The proposed construction admits a biological
interpretation of the synaptic components in terms of excitation,
inhibition, and homeostatic regulation. Moreover, we frame the canonical
prescription of Dayan~\&~Abbott \citep[Section~7.4]{AD:05} as a special
case of our synaptic matrix construction. Additionally, we show that the
emergence of “anti-memories” is possible only for pathological cases
reducible to the use of Hopfield-type synaptic matrix, and explore the
existence of spurious equilibria, particularly homogeneous ones. The second
theorem goes on to establish sufficient conditions for the local asymptotic
stability of the rescaled prototypical memories and, leveraging results
from Grossberg \citep{G:83} and Hopfield \citep{H:84}, we proceed by
defining an energy function to analyze the global behavior of
trajectories. Finally, we investigate numerically the tightness of these
stability conditions and visualize the energy landscapes for two relevant
examples. Notably, simulations reveal how the choice of a negative
homeostatic strength results in wider stability regions over the space of
parameters, compatibly with the canonical sign choice for the homeostatic
term found in the literature. To enhance readability, the proofs of the
technical results are deferred to the Appendix.

Notation: We let $\real^{n\times m}$ denote the set of $n\times m$ matrices with real entries. The symbol $\vectorones[n]$ indicates an $n$-dimensional vectors of ones and $I_n$ the $n\times n$ identity matrix. Given a matrix $A\in\real^{n\times m}$, $A^\top$ is the transpose of $A$. For a symmetric matrix $A=A^\top$, we write $A\succ 0$ ($A\succeq 0$) if $A$ is positive definite (positive semidefinite, respectively). Moreover for two symmetric matrices $A,B$, we write $A\succ B$ ($A\succeq B$) if $A-B\succ 0$ ($A-B\succeq 0$). Given a vector $x\in\real^n$, $\diag(x)$ is the diagonal matrix with the entries of $x$ as diagonal entries. 
If $f(x)$ is a real-valued function, $f'(x)$ denotes the derivative of $f$. A function is weakly increasing if $f(x_{1})\le f(x_2)$ for all $x_{1}, x_2\in\real$ with $x_{1}<x_2$ and strictly increasing if $f(x_{1})< f(x_2)$ for all $x_{1}, x_2\in\real$ with $x_{1}<x_2$. 

\section{Voltage vs.\ firing rate models}
The \emph{voltage} and \emph{firing rate} models are two
widely used neural network models for associative memory. In its classic
autonomous version, the continuous-time \emph{voltage} model is
\begin{align}\label{eq:H}\tag{V}
    \dot \volt(t) = -\volt(t) + W\Psi(\volt(t)),
\end{align}
where $\volt(t)$ is an $n$-dimensional, time-dependent vector containing
voltages of neuronal populations, $\dot \volt(t)$ is its time derivative,
$W\in\real^{n\times n}$ is the synaptic matrix, and $\Psi\colon \real^{n}
\to \real^{n}$ denotes the activation function.  The activation function
$\Psi(\cdot)$ is typically assumed to be diagonal and homogeneous, meaning
that $\Psi(\volt)_i = \psi(\volt_i)$, where $\psi\colon \real \to \real$ is
a scalar function applied entrywise to $\volt \in \real^n$.  Moreover,
$\psi(\cdot)$ is typically assumed to be weakly increasing and odd.  The
prototypical memories of interest are binary vectors encoded in the
synaptic matrix $W$, by means of one-shot Hebbian learning \citep{H:82,
  GK:02}. Scaled version of these memories are subsequently retrieved as
the stable equilibria of equation \eqref{eq:H}.
  
If, in addition to the previous assumptions, $W$ is symmetric and the
activation function $\psi(\cdot)$ is differentiable, bounded and strictly
increasing, then the trajectories of \eqref{eq:H} are guaranteed to
converge to the set of its equilibria. The classic proof of this result
resorts to the LaSalle invariance principle \citep{K:02} and to the
associated energy function \citep{H:84}:
\begin{align}\label{eq:E}
  \subscr{E}{H}(\volt) =& -\frac{1}{2} \Psi(\volt)^\top W \Psi(\volt) +
  \sum_{i=1}^n \int_0^{\psi(\volt_i)} \psi^{-1}(z)\, \de z.
\end{align}
Indeed, it can be shown that along the system’s trajectories, the time
derivative of the energy function is
\begin{equation}\label{eq:DE}
    \subscr{\dot{E}}{H}(\volt)=-\dot{\volt}^{\top}\diag(\Psi'(\volt))\dot{\volt},
\end{equation}
which is strictly negative for all $\volt$ such that $\dot{\volt}\neq 0$.

In contrast, the autonomous continuous-time \emph{firing rate} model reads
as
\begin{align}\label{eq:FR}\tag{FR}
    \dot x(t) = -x(t) + \Phi(W x(t)), 
\end{align}
where the $n$-dimensional state vector $x$ contains the firing rates of the
neuronal populations, and $\dot x(t)$ denotes its time derivative.  The
activation function $\Phi(\cdot)$ is typically assumed to be diagonal and
homogeneous, meaning $\Phi(x)_i = \phi(x_i)$, where $\phi\colon \real \to
\real$ is a scalar, weakly increasing function.  In \emph{firing rate}
systems, the activation function is assumed to be non-negative, meaning
$\phi(x)\ge 0$ for all $x\in\real$. The non-negativity of $\phi(\cdot)$
implies that \eqref{eq:FR} is a \emph{positive system}, meaning that the
trajectories of \eqref{eq:FR} do not leave the positive orthant for
non-negative initial conditions: if each entry of the initial condition
satisfies $x_i(0)\geq0$, then each entry remains non-negative for all time.

Finally, the relationship between the \emph{voltage} and \emph{firing rate}
models has been investigated in prior work \citep{MF:12, F:18}. These
studies demonstrate that the two models\footnote{In these works, the
\emph{voltage} model is presented with dynamical equations $\dot v = -v +
W\Phi(v)$, while the \emph{firing rate} model is expressed as $\dot r = -r
+ \Phi(Wr)$. Both models are grouped under the umbrella of firing rate
systems, assuming they share a common positive activation function $\Phi()$
and the same synaptic matrix $W$.} can be made equivalent through
appropriate transformations of state and input. However, they do not
address the challenge of designing the synaptic matrix $W$ to achieve
desired memory patterns as stable equilibria within the system. Moreover,
when $W$ is low-rank, as is typical in associative memory contexts, the
relationship between the models becomes less direct and more
nuanced. Specifically, the equivalence is derived by projecting the
\emph{firing rate} dynamics onto the subspace spanned by the columns of
$W$, a constraint absent in the original dynamics. In addition, bridging
the dynamics of the two models introduces a time-varying external input,
complicating the use of standard Lyapunov methods to establish convergence
to equilibria.

\section{Equilibria assignment through synaptic weights}\label{sec:equilibria}

\subsection{Design techniques}

Consider the \emph{firing rate} model \eqref{eq:FR}. From now on we will assume that the activation function satisfies the following hypothesis.

\begin{assumption}[Positive and monotonic activation functions]\label{assum:activation-function}
The activation function $\Phi(\cdot)$ in \eqref{eq:FR} is diagonal and
homogeneous, i.e., there exists $\phi\colon\real\to\real$ such that
$\Phi(x)_i= \phi(x_i)$ for all $x\in\real^n$. Moreover, the function
$\phi(\cdot)$ is continuous, non-negative, and weakly increasing.
\end{assumption}

Given a set of $\{0,1\}$-valued vectors in $\real^n$
\begin{equation}\label{eq:PM}
  \{\xi^{\mu}\}_{\mu=1}^P 
\end{equation}
corresponding to \emph{prototypical memory} patterns with either active or inactive units, we are interested in designing biologically plausible synaptic matrices $W$ such that an appropriately scaled version of these vectors are equilibria of the \emph{firing rate} dynamics \eqref{eq:FR}. 
We consider prototypical memories $\{\xi^{\mu}\}_{\mu=1}^P$ satisfying the following hypothesis.

\begin{assumption}[Equally sparse and correlated memories]\label{assum:memories}
  Given an \emph{average activity parameter} $p\in (0,1)$, the prototypical
  memories $\{\xi^{\mu}\}_{\mu=1}^P$, $\xi^{\mu}\in\{0,1\}^n$ satisfy
  \begin{subequations}\label{eq:memories}
    \begin{align}
      &\text{(equal sparsity):}\qquad
      \vectorones[n]^\top \xi^{\mu}=p n , \ \ \ &\forall\, \mu\in\until{P}, \label{eq:memories-b}\\
      &\text{(equal correlation):}\qquad {\xi^{\mu}}^\top \xi^{\nu} = p^{2} n, \ \ \
      &\forall\, \label{eq:memories-c} \mu\ne \nu.
    \end{align}
  \end{subequations}
\end{assumption}

The conditions in Assumption~\ref{assum:memories} are inspired by the
choice of prototypical memories in the classic Hopfield
\citep{HW:87,XSZ:13} and \emph{firing rate} models \citep[Section~7.4]{AD:05}. In
probabilistic terms, the parameter $p$ is related to the average activity
of the network in each memory pattern assuming that these activities are
uncorrelated. The equal sparsity \eqref{eq:memories-b} and equal
correlation \eqref{eq:memories-c} constraints refer to the simplest
statistical structure that memory patterns can have. As a matter of fact, if we assume that the entries of prototypical memories $\xi^{\mu}_h\in\{0,1\}$, $\mu=1,\ldots,P$, $h=1,\ldots,n$ are i.i.d. random binary variables such that the probability that $\xi^{\mu}_h=1$ is equal to $p$ for all $h,\mu$,
then the conditions \eqref{eq:memories-b} and \eqref{eq:memories-c} are satisfied in expectation.  However, as described in the first section of the Appendix, our analysis applies to a more general scenario where the memory patterns have more general correlations.
        
We now provide two useful definitions.

\paragraph{From prototypical memories to synaptic matrices.}
To a set of prototypical memories $\{\xi^{\mu}\}_{\mu=1}^{P}$ satisfying
the sparsity and correlation Assumption~\ref{assum:memories} with average
activity $p$, we associate an $n\times{n}$ dimensional
\emph{covariance-based synaptic matrix}
\begin{equation}\label{eq:W-map}
  W ~:=~ \frac{\alpha}{p(1-p)n}\sum_{\mu=1}^{P}(\xi^{\mu} - p
  \vectorones[n])(\xi^{\mu} - p \vectorones[n])^\top ~+~ \frac{\gamma}{
    n}\vectorones[n]\vectorones[n]^\top,
\end{equation}
defined by the following parameters:
\begin{itemize}
\item the \emph{correlation strength} $\alpha\in\real$ which modulates the
  covariance part of $W$ given by the outer products of the
  shifted prototypical memories, and
\item the \emph{homeostatic strength} $\gamma\in\real$ which controls the
  magnitude of the component of the covariance-based synaptic matrix that modulates the
  synaptic response proportionally to the global network activity. This
  contribution to the global activity is referred to as homeostatic from
  the electrochemical balancing mechanism observed in real neuronal
  networks. Typically, homeostatic terms are discussed in the literature
  \citep{GRS:90, AD:05} as global inhibitory terms, with homeostatic
  strength $\gamma\le 0$.
\end{itemize}

\paragraph{From prototypical memories and selected firing rates to retrievable
  memories.}

Select low and high firing rates $x_{0}<x_{1} \in\real$ both belonging to the range of the activation function $\phi$,
and associate to each prototypical memory $\xi^\mu \in\real^n$ a \emph{retrievable memory}
$\bar\xi^\mu \in\real^n$ defined by
\begin{equation}\label{eq:RM-components}
  \bar\xi^{\mu}_i =
  \begin{cases}
    x_{0}, \qquad &\text{if }  \xi^{\mu}_{i} =0, \\
    x_{1}, \qquad &\text{if }  \xi^{\mu}_{i} =1.
  \end{cases}
\end{equation}
In vector format, $\bar\xi^\mu ~:=~ (x_{1}-x_{0})\xi^{\mu}+x_{0}\vectorones[n]$.

The firing rates $x_{0}$ and $x_{1}$ are the neural manifestation of the activity of population of neurons experiencing different input currents \citep{VKA:05}. The input currents are the incoming activity that each neuron processes and filters by means of its activation function. Therefore, we can mathematically relate input currents and firing rates in the following way. Let $\I_{0}\in\real$ represent a weak input current, and $\I_{1}\in\real$ a strong input current, so that $\I_{0}<\I_{1}$. Then the states of the neurons associated with low and high firing rates are $x_{0}=\phi(\I_{0})$ and $x_{1}=\phi(\I_{1})$, respectively.
       
Given the definitions of covariance-based synaptic matrix and retrievable memories, the
following theorem characterizes when the retrievable memories are equilibria
of the \emph{firing rate} system \eqref{eq:FR}. We postpone the proof to the
Appendix.

\begin{theorem}[Equilibria assignment through covariance-based synaptic weights]
  \label{thm:W}
  Consider the \emph{firing rate} model~\eqref{eq:FR} with activation functions
  satisfying the positivity and monotonicity
  Assumption~\ref{assum:activation-function}.  Consider prototypical
  memories $\{\xi^{\mu}\}_{\mu=1}^P$ satisfying the sparsity and
  correlation Assumption~\ref{assum:memories} with average activity $p$.

  If there exist currents $\I_{0}<\I_{1}$, with corresponding
  low and high firing rates $x_{0}:=\phi(\I_{0})$,
  $x_{1}:=\phi(\I_{1})$, and
  correlation and homeostatic strengths $\alpha$ and $\gamma$ satisfying
  \begin{subequations}\label{eq:parameters}
    \begin{align}
      &\alpha = \frac{\I_{1}-\I_{0}}{x_{1}-x_{0}},\label{eq:alpha}\\
      &\gamma = \frac{p \I_{1}+(1-p)\I_{0}}{p x_{1}+(1-p)x_{0}}, \label{eq:gamma}
    \end{align}
  \end{subequations}
  then each retrievable memory $\bar\xi^\mu$ is an equilibrium point of the
  \emph{firing rate} model \eqref{eq:FR} with 
  the covariance-based synaptic matrix~\eqref{eq:W-map}.
\end{theorem}

Some comments on the parameters $\alpha$ and $\gamma$ are in order. First, in \eqref{eq:alpha} the parameter $\alpha$ is positive and is inversely proportional to the slope of the straight line intersecting the activation function $\phi(\cdot)$ at the coordinates $\I_{0}$ and $\I_{1}$.  
Geometrically, the inverse of the parameter $\gamma$ is the slope of the line that passes through the origin and intersects the straight line $x=\alpha^{-1}(\I-\I_{0})+x_{0}$ in the coordinate $\I_{p}=p\I_{1}+(1-p)\I_{0}$ (see Figure \ref{fig:gamma}). 
\begin{figure}[h!]
    \centering 
    \includegraphics[width=.55\columnwidth]{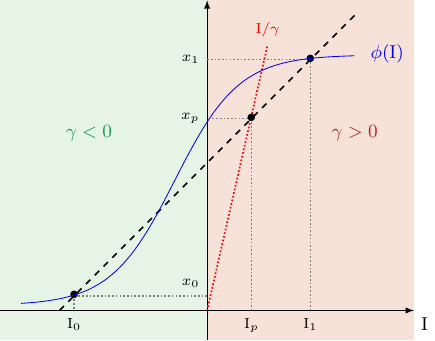}
    \caption{Graphical interpretation of the parameter $\gamma$ in \eqref{eq:gamma}: $\gamma^{-1}$ coincides with the slope of the straight line intersecting the origin and the point $(I_{p},x_{p})$, where $\I_{p}:=p \I_{1}+(1-p)\I_{0}$ and $x_{p}:=p x_{1}+(1-p)x_{0}$. Note that, if the point $(\I_{p},x_{p})$ belongs to the green area, then $\gamma$ is negative, whereas if the point $(\I_{p},x_{p})$ belongs to the red area, then $\gamma$ is positive.}
    \label{fig:gamma}
\end{figure}

\subsection{The canonical case of Dayan \& Abbott}
The matrix in \eqref{eq:W-map} closely resembles the one in \citep[Section~7.4]{AD:05}, which is the canonical choice throughout the theoretical neuroscience literature. We next show that the latter is indeed a special case of \eqref{eq:W-map}. 
The construction presented in \citep{AD:05}
assumes the following:
\begin{enumerate}
    \item[(i)] there exists $z^*$ such that $\phi(z)=0$ for $z\le z^*$;
    \item[(ii)] there are scalars $\lambda,\delta>0$ such that $\I_{0}=-\delta(1+p\lambda)\le z^*$, $\I_{1}=\delta(\lambda-1-p\lambda)$ and $x_{1}:=\phi(\I_{1})=\delta$. Since it is assumed that $\I_{0}\le z^*$, then $x_{0}:=\phi(\I_{0})=0$.
\end{enumerate}
From the previous relations it can be seen that $\I_{0}$, $\I_{1}$ cannot be chosen arbitrarily. Indeed, the above defined $\I_{0}$ and $\I_{1}$ have to satisfy the equation $p\I_{1}+(1-p)\I_{0}+\phi(\I_{1})=0$. This implies that for this type of construction we can start from any $\phi(\cdot)$ satisfying condition (i) and any real number $\I_{1}$ such that 
$$\frac{p\I_{1}+\phi(\I_{1})}{p-1}\le z^*$$ 
and take $\I_{0}:=\frac{p\I_{1}+\phi(\I_{1})}{p-1}$.
It can be seen that the parameters $\alpha,\gamma$ in this case become
\begin{align*}
    &\alpha=\frac{\I_{1}-\I_{0}}{x_{1}}=\lambda,\\
    &\gamma= -\frac{1}{p}
\end{align*}
and the covariance-based synaptic matrix then becomes
\begin{align}\label{eq: WAbb}
   W &= \frac{\lambda}{p(1-p)n}\sum_{\mu=1}^{P}(\xi^{\mu}-p\vectorones[n])(\xi^{\mu}-p\vectorones[n])^{\top}-\frac{1}{pn}\vectorones[n]\vectorones[n]^{\top}\nonumber
   \\
    & = \frac{\lambda}{p(1-p)\delta^{2}n}\sum_{\mu=1}^{P}(\bar\xi^{\mu}-p\delta\vectorones[n])(\bar\xi^{\mu}-p\delta\vectorones[n])^{\top}-\frac{1}{pn}\vectorones[n]\vectorones[n]^{\top}
\end{align}
where from the first to the second line in Equation (\ref{eq: WAbb}) we have multiplied and divided by $\delta^{2}$ the terms in the summation and used the fact that $\bar\xi^{\mu}=\delta\xi^{\mu}$. We have obtained in this way exactly the expression of $W$ given in  \citep{AD:05}.

\subsection{A bridge from math to biology}

We now want to address the biological interpretation of the map (\ref{eq:W-map}) and gain some insight on how the model can capture different aspects of neural processing. As it stands, the positivity of the \emph{firing rate} model offers a valuable interpretative tool to bridge dynamical system theory and neural processes. Specifically, it offers the possibility to understand the role of the different components of the covariance-based synaptic matrix and how they interact with the network activity to generate neuronal rates of firing. Separating the different components 
\begin{multline}
    W ~=~ \underbrace{\frac{\alpha}{p(1-p) n}\sum\nolimits_{\mu}^{P}\xi^{\mu}{\xi^{\mu}}^{\top}}_{\text{excitatory correlation network }W^{\text{ex}}} \\~-~\underbrace{\frac{\alpha}{(1-p) n}\left(\sum\nolimits_{\mu}^{P}\vectorones[n]{\xi^{\mu}}^{\top}+\xi^{\mu}\vectorones[n]^{\top}\right)}_{\text{inhibitory memory-network interaction }W^{\text{in}}}
     ~+~\underbrace{\left(\frac{\alpha}{(1-p) n}Pp+\frac{\gamma}{ n}\right)\vectorones[n]\vectorones[n]^{\top}}_{\text{global homeostatic network }W^{\text{om}}},
\end{multline}
we observe that 
\begin{itemize}
    \item $W^{\text{ex}}$ is a relatively sparse excitatory network that takes into account how specific neuronal sub-clusters positively excite to produce fixation of the activity on a given pattern. Indeed, from the outer product of the patterns, we have that $W^{\text{ex}}_{ij}\neq{0}$ only if there exists $\mu=1,\dots,P$ such that $\xi^{\mu}_{i}=1$ and $\xi^{\mu}_{j}=1$.
    \item $W^{\text{in}}$ is the inhibitory network and is the sum of two terms, both regulating the global network activity with weights proportional to the parameters $\alpha,p$. The first term, given by the outer products $\vectorones[n]{\xi^{\mu}}^{\top}$, provides global inhibition proportional to the correlation of the network activity and the different memory patterns. The second term, given by the outer products $\xi^{\mu}\vectorones[n]^{\top}$, provides selective inhibition to the units of the memory patterns proportional to the summed network activity.
    \item $W^{\text{om}}$ is the homeostatic network and it exercises global regulation of the \emph{firing rate} network via stimulation that depends on the normalized network rate
    \begin{equation}
        s(t) = \frac{\vectorones[n]^{\top}x(t)}{n}
    \end{equation}
\end{itemize}

Elaborating further on the biological constraints that lead to the construction of the covariance-based synaptic matrix, we observe that the $W$ in \eqref{eq:W-map} depends explicitly from the prototypical memory vectors $\{\xi^{\mu}\}_{\mu=1}^P$. Notably, the form of $W$ resembles the synaptic matrix of the classic Hopfield model \citep{H:82}, which can be constructed using the biologically plausible Hebbian learning rule \citep{GK:02}. Hebbian learning is a co-variation learning rule based on the principle articulated by D.~O.~Hebb ``cells that fire together, wire together'' \citep{H:49}, implying the use of local information to alter the strength of the synaptic couplings. In our case, we have that the covariance-based synaptic matrix is simply given by a one-shot Hebbian learning rule. Hebbian learning as a learning framework adheres to the biological plasticity processes known as long-term potentiation (LTP) and long-term depression (LTD) \citep{BGM:73, MP:04}, known to be responsible for learning and forgetting in mammals. These local synaptic modifications aim to stabilize the connection strengths between neurons, enabling the retrieval of certain patterns through the collective dynamics of the network.

\subsection{When the antimemories are equilibria}

Consider a set of prototypical memories \eqref{eq:PM} satisfying the equal sparsity and the equal correlation constraints of Assumption \ref{assum:memories} and let $W$ be as in Theorem \ref{thm:Wgen}. Assume that $\phi(\cdot)$ in \eqref{eq:FR} satisfies Assumption \ref{assum:activation-function}. For any prototypical memory $\xi$ we define the corresponding antimemory as $\xi^\text{ant}:=\vectorones[n]-\xi$ that is the vector in which we exchange the zeros with ones and vice versa. We wonder whether $\bar\xi^\text{ant}:=(x_{1}-x_{0})\xi^\text{ant}+x_{0}\vectorones[n]$ is an equilibrium point of \eqref{eq:FR}.
Notice that this is exactly what happens for the Hopfield models. The answer to this question is given in the Appendix where it is shown that anti-memories provide equilibrium points for \eqref{eq:FR} if and only if 
\begin{equation}\label{eq:condantim}
    p=1/2\quad \text{or}\quad \I_{0}x_{1}=\I_{1}x_{0}.
\end{equation}
Notice that, while for the Hopfield model we have that $\I_{0}x_{1}=\I_{1}x_{0}$ and so anti-memories are equilibrium points regardless the value of $p$, this is not true for \emph{firing rate} models since the condition $\I_{0}x_{1}=\I_{1}x_{0}$  generally fails. 

\subsection{Existence of homogeneous equilibria}\label{subsec:hom}

A well-known issue in associative memory networks with a Hebbian synaptic
matrix is the emergence of spurious equilibria \citep{AHM:85,H:18}, which
are equilibria that do not correspond to the intended memories.
We show now how the covariance-based synaptic matrix invariably produces at least one
spurious equilibrium point with equal entries. Such equilibria will be called \emph{homogeneous equilibria}.

\begin{lemma}[Homogeneous equilibria]\label{lm:hom}
  With the same notation and under the same assumptions as
  Theorem~\ref{thm:W},
\begin{enumerate}
\item there exists at least one solution $z$ to the equation
  \begin{equation}\label{eq:puntfiss}
    \gamma^{-1} z = \phi(z),
  \end{equation}
\item for each solution $z$, the vector $\bar x = \gamma^{-1} z
  \vectorones[n]$, is an equilibrium point for the \emph{firing rate} model \eqref{eq:FR}.
\end{enumerate}
\end{lemma}

Equation \eqref{eq:puntfiss} has a simple graphical interpretation. Indeed, its solutions result from the intersections between the graph of $\phi(z)$ and the straight line $\gamma^{-1} z$. 
Observe that in case $\gamma<0$, then $f(z):=\phi(z)-\gamma^{-1} z$ is strictly increasing and such that $f(-\infty)=-\infty$ and $f(+\infty)=+\infty$.
Hence there always exists exactly one solution to equation \eqref{eq:puntfiss}. If instead $\gamma>0$, then there might be multiple solutions (see Figure \ref{fig:exspurious}).

\begin{figure}[h!]
    \centering
    \subfloat[$\gamma$ large]{\includegraphics[width=.5\columnwidth]{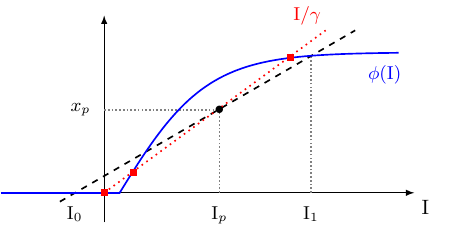}}
    \subfloat[$\gamma$ small]{\includegraphics[width=.5\columnwidth]{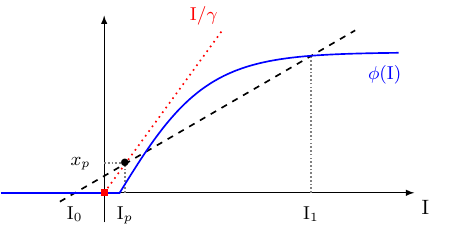}}

    \caption{Graphical illustration of the existence of homogeneous equilibria as $\gamma$ varies. The red squares denote the solutions of equation \eqref{eq:puntfiss}, each associated with a homogeneous equilibrium point.}
    \label{fig:exspurious}
\end{figure}

\section{Stability analysis}

In this section, we examine the stability of equilibria of the \emph{firing rate} dynamics \eqref{eq:FR} with synaptic matrix $W$ constructed as in the previous section.

\subsection{On the local stability of retrievable memories}

A sufficient condition that ensures local asymptotic stability of an
equilibrium point can be derived via Lyapunov's indirect method
\citep[Theorem 4.7]{K:02}. Specifically, if the Jacobian matrix of \eqref{eq:FR}
evaluated at an equilibrium point $\bar x$, namely
\begin{equation}
J(\bar x) = -I+\diag(\Phi'(W\bar{x}))W,\label{eq:Jacobian}
\end{equation}
has all its eigenvalues with strictly negative real part, then $\bar{x}$ is
locally asymptotically stable.

Building on this result, we next establish a sufficient condition for the
local stability of the retrievable memories as equilibria of \eqref{eq:FR}.
The proof is postponed to Appendix.

\begin{theorem}[Local stability condition]\label{thm:stab}
   With the same notation and under the same assumptions as
    Theorem~\ref{thm:W}, if
  \begin{equation}\label{eq:stability}
    \max\bigl\{\phi'(\I_{0}),\phi'(\I_{1})\bigr\}\max\{\alpha,\gamma\}<1,
  \end{equation} 
  then each retrievable memory $\bar\xi^\mu$ is locally asymptotically
  stable for the \emph{firing rate} model~\eqref{eq:FR}.
\end{theorem}

\begin{remark}[The canonical case of \cite{AD:05}]
For the memories construction presented in \cite[Section~7.4]{AD:05}, since $\gamma=-\frac{1}{p}<0$, $\alpha=\lambda>0$, and $\phi'(\I_{0})=0$, the stability condition \eqref{eq:stability} is simply
$$\phi'(\I_{1})<\lambda^{-1}.$$ 
\end{remark}

Theorem \ref{thm:stab} indicates that the derivative of the activation function $\phi(\cdot)$ at the input coordinates $\I_{0}$ and $\I_{1}$ plays an important role in the stability of the retrievable memory patterns $\{\bar\xi^{\mu}\}_{\mu=1}^P$. Specifically, the smaller are $\phi'(\I_{0})$ and $\phi'(\I_{1})$, the smaller is the left-hand side of \eqref{eq:stability}, suggesting that the retrievable memories are more likely to be stable. In addition, the slope of the straight line intersecting the activation function at $\I_{0}$ and $\I_{1}$ also affects the condition \eqref{eq:stability} through parameter $\alpha$. 
In particular, when $\gamma\le \alpha$, stability is guaranteed if both $\phi'(\I_{0})<1/\alpha$ and $\phi'(\I_{1})<1/\alpha$, which means that the straight line intersecting the activation function at the points $(\I_{0},x_{0})$ and $(\I_{1},x_{1})$ has to intersect it from below (see Figure \ref{fig:stability}).

\begin{remark}[The case $\gamma\le\alpha$]
We have seen that when $\gamma\le\alpha$, the stability condition simplifies to $\max\bigl\{\phi'(\I_{0}),\phi'(\I_{1})\bigr\}<\alpha^{-1}$. It can be proved that $\gamma\le\alpha$ holds if and only if $\I_0x_1\le \I_1 x_0$. Observe that this condition always holds if $\I_0\le 0\le \I_1$.
\end{remark}

\begin{figure}
    \centering
    \includegraphics[width=1\linewidth]{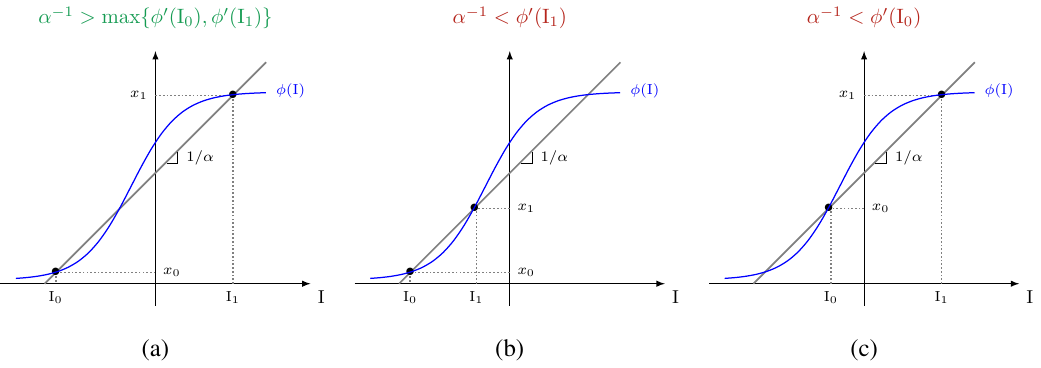}
    \caption{Graphical interpretation of the local stability condition \eqref{eq:stability} when $\gamma\le \alpha$. In panel (a) condition \eqref{eq:stability} is satisfied, while in panels (b)-(c) condition \eqref{eq:stability} is not satisfied and hence Theorem \ref{thm:stab} cannot be applied. }
    \label{fig:stability}
\end{figure}

\begin{remark}[Instability condition]
  Techniques similar to those in the proof of Theorem \ref{thm:stab} lead
  to the following statement: if
  \begin{equation} \label{eq:instability}
  \max\{\phi'(\I_{0})[p\alpha+(1-p)\gamma],\phi'(\I_{1})[(1-p)\alpha+p\gamma]\}>1,
  \end{equation}
  then the equilibria $\{\bar{\xi}^{\mu}\}_{\mu=1}^P$ of \eqref{eq:FR} are
  unstable.
\end{remark}

\begin{remark}[homogeneous equilibria]
The local stability of homogeneous equilibria introduced in Section \ref{subsec:hom} can be analyzed with similar arguments used for proving Theorem \ref{thm:stab}. Specifically, if $\bar x=\gamma^{-1} z\vectorones[n]$ an equilibrium point for \eqref{eq:FR}, where $z$ is a real
number satisfying equation \eqref{eq:puntfiss}, then it is easy to see that $\bar x$ is locally stable if
\begin{equation*}
  \phi'(z)\max\{\alpha,\gamma\}<1.
\end{equation*}
It can also be proved that such equilibrium point is unstable if instead
\begin{equation*}
  \phi'(z)\max\{\alpha,\gamma\}>1,
\end{equation*}
In the special case in which $\gamma\le\alpha$,
the stability condition simplifies to $\phi'(z)<\alpha^{-1}$. 

\end{remark}

\subsection{On the global stability of retrievable memories}

We now present a result on the global behavior of the trajectories of \eqref{eq:FR} based on an~energetic characterization of the \emph{firing rate} model.

The function
\begin{equation}\label{eq:energy-FR}
    \subscr{E}{FR}(x)\! =\! -\frac{1}{2} x^\top W x + \sum_{i=1}^n \! \int_0^{x_i}\!\!\! \phi^{-1}(z)\, \de z
\end{equation}
will serve as energy for the \emph{firing rate} model \eqref{eq:FR}, where $\phi^{-1}$ denotes any right inverse of $\phi$.\footnote{A right inverse of a function $f\colon X\to Y$ is a function $g\colon Y\to X$ such that $f(g(y))=y$ for all $y\in Y$.} Note that \eqref{eq:energy-FR} coincides with the classic Hopfield energy \eqref{eq:E} expressed in the variable $y=\Phi(x)$. The proof of the following theorem can be found in the Appendix.

\begin{theorem}[Global convergence to equilibria]\label{thm:conv}
With the same notation and under the same assumptions as Theorem~\ref{thm:W}, assume that the activation function $\phi(\cdot)$ takes value in a bounded interval $\mathcal{I}$.  If $x(0)\in\mathcal{I}^n$, then $\subscr{\dot{E}}{FR}\le 0$ along the flow of \eqref{eq:FR} and each trajectory of \eqref{eq:FR} converges to the set of equilibrium points of~\eqref{eq:FR}.
\end{theorem}

\section{Illustrative examples}

The aim of this section is to provide hints to the reader on the design choices that result in effective associative memory \emph{firing rate} networks. Choosing the proper combination of parameters to set up an effective \emph{firing rate} model can be challenging due to the presence of five free parameters; namely, the average activation $p$, the input currents $\I_{0}$, $\I_{1}$, the \emph{gain factor} $\rho$ and the \emph{activation current} $\I^{\star}$. The latter two parameters were not previously discussed, but as we will see while we study two relevant examples, they are of critical importance in determining the stability properties of the system. Therefore, we will begin by fixing three out of the five free parameters, and study the stability properties of the system as the other two varies. Specifically, taking inspiration from neurobiological reality \citep{WG:14} we will fix the average neural activation as $p=0.2$. We will then continue with the definition of two paradigmatic examples of activation function, and numerically study the stability properties of the system as the slope and offset parameters vary. The stability condition \eqref{eq:stability} is independent of network size; instead, the numerical stability of the \emph{firing rate} model is studied for a network with $n=1000$ and $P=6$ (see the Appendix for details on the construction of the deterministic memories). We considered the \emph{firing rate} system to have numerically stable retrievable memories when the Jacobian of the vector field \eqref{eq:FR} evaluated at each of the memories has all eigenvalues with negative real part. Additionally, we will present the energy function associated to specific \emph{firing rate} models where the slope and offset parameters have been fixed to values that ensure either the stability or instability of the memories as equilibria for the system. In order to do so, we will evaluate the energy on a mesh interpolating the space between two memories as
\begin{align}
    x = t_{1}\xi^{1}+t_{2}\xi^{2}, \quad t_{1},t_{2}\in[0,1].
\end{align}
The energy profiles are plotted only for values of $t_{1}, t_{2}$ such that $x\in[0,1]^{n}$. 

Finally, we test the retrieval performances of the \emph{firing rate} model using the average overlap parameter
\begin{equation}
    s^{\nu}(t)=\frac{x(t)^{\top}\xi^{\nu}}{pn}
\end{equation}
where the average is taken over all the $\nu=1,\dots,P$. In this testing phase, the prototypical memories $\{\xi^{\mu}\}_{\mu=1}^{P}$ are drawn randomly so that the equal sparsity and equal correlation constraints are satisfied in expectation. Each retrieval is characterized by dynamics initialized arbitrarily close to the pattern $\bar \xi^{\nu}$ and, if the memory is a stable equilibria of the system, then by the equal sparsity constraint there exists a $\bar t>0$ such that $s^{\nu}(t)\approx{1}$ for all $t>\bar t$. Conversely, by the equal correlation constraint we would expect that $s^{\mu}(t)\approx{p}$ for all $t>\bar t$ and $\mu\neq\nu$.

\subsection{Rectified activation functions}
Consider the rectified hyperbolic tangent activation (Fig.~\ref{fig:rect-tanh})
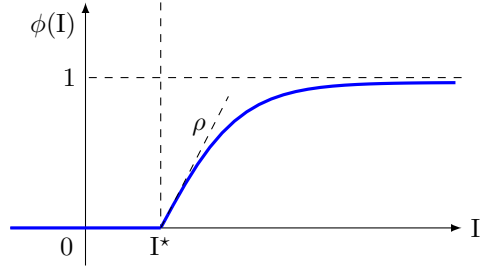
\begin{figure}[!h]
\begin{center}
\begin{tikzpicture}
\draw[-latex] (-2,0) -- (4,0) node[right] {$\I$};
\draw[-latex] (-1,-0.5) -- (-1,3) node[below left] {$\phi(\I)$};
\draw[dashed] (0,3) -- (0,0) node[below]{$\I^{\star}$};
\draw[dashed] (4,2) -- (-1,2) node[left]{$1$};
\node at (-1.25,-0.25) {$0$};
\draw[blue,very thick] (-2,0) -- (0,0);
\begin{axis}[
    height=4cm,
    width=5.5cm,
    xmin=0, xmax=5,
    ymin=0, ymax=1.25,
    axis lines=center,
    axis on top=true,
    hide axis,
    domain=0:5,
    ]
    \addplot [mark=none,draw=blue,very thick] {tanh(0.75*\x)};
\end{axis}
\draw[dashed] (0,0) -- (0.9,1.75) node[below left=0.25cm]{$\rho$};
\end{tikzpicture}
\end{center}
\caption{Rectified hyperbolic tangent activation function \eqref{eq:rect-tanh}.}
\label{fig:rect-tanh}
\end{figure}
\begin{align}\label{eq:rect-tanh}
\phi(\I) = \begin{cases}\tanh(\rho (\I-\I^{\star})), & x> \I^{\star},\\ 0, & x\le \I^{\star},\end{cases}
\end{align}
where the parameter $\I^{\star}\in\real$ is the largest input current that yields a zero output and will be referred to as  \emph{activation current}, and $\rho\in\real$, $\rho>0$, equals the maximal derivative of $\phi(\cdot)$ and will be termed \emph{gain factor} (see also Fig.~\ref{fig:rect-tanh}). This function is typically used to replicate the \emph{firing rate} response to constant input currents observed in biological neural networks \citep[Ch.~2]{AD:05}. In Fig.~\ref{fig:PD2}(a,d) we study the stability properties of the system as a function of the parameters $\rho$, $\I^{\star}$, and the sign of the homeostatic strength $\gamma$. As observable, a negative homeostatic strength $\gamma$, which coincides with inhibitory feedback, results in a wider stability region for the \emph{firing rate} system. Figures \ref{fig:PD2}(b) and \ref{fig:PD2}(e) plot the energy associated with \emph{firing rate} models with, $\rho=4.8$, $\I^{\star}=0.2$ for stable retrievable memories
(blue square marker in Fig.~\ref{fig:PD2}(a)) and $\I^{\star}=0.8$ for unstable retrievable memories (red starred marker in Fig.~\ref{fig:PD2}(a)). As observable from Fig.~\ref{fig:PD2}(c), when the network parameters ensure the stability of the memories, the \emph{firing rate} model is capable of performing the correct retrieval of the intended memory. Instead, when the network parameters lie in the instability region, the \emph{firing rate} model activity collapses to a state of almost total inactivation (Fig.~\ref{fig:PD2}(f)), compatible with the energy profile in Fig.~\ref{fig:PD2}(e).

\subsection{Sigmoidal activation functions}
Consider the sigmoidal activation function (Fig.~\ref{fig:sigm})
\begin{figure}[!tbph]
\begin{center}
\begin{tikzpicture}
\draw[-latex] (-2.1,0) -- (4,0) node[right] {$\I$};
\draw[-latex] (-0.5,-0.5) -- (-0.5,3) node[below left] {$\phi(\I)$};
\draw[dashed] (4,2) -- (-0.5,2) node[left]{\footnotesize $1$};
\draw[dashed] (4,1) -- (-0.5,1) node[left]{$\frac 1 2$};
\node at (-0.75,-0.25) {$0$}; \node at (-0.15,-0.25) {$\I^{\star}$};
\begin{scope}[xshift=-2cm,yshift=0.01cm]
\begin{axis}[
    height=4cm,
    width=7.5cm,
    xmin=-5, xmax=5,
    ymin=-0.01, ymax=1.25,
    axis lines=center,
    axis on top=true,
    hide axis,
    domain=-5:5,
    ]
    \addplot [mark=none,draw=blue,very thick] {1/(1+exp(-x))};
\end{axis}
\draw[dashed] (1.75,0) -- (4.25,2);
\draw[dashed] (3,0) -- (3,3);
\node at (3.3,1.55) {\small $\rho$};
\end{scope}
\end{tikzpicture}
\end{center}
\caption{Sigmoidal activation function \eqref{eq:sigm}.}
\label{fig:sigm}
\end{figure}
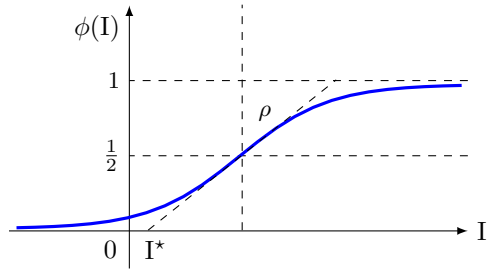
\begin{align}\label{eq:sigm}
\phi(\I) = \frac{1}{1+e^{-4\rho (\I-\I^{\star}-(2\rho)^{-1})}}, \ \ \ \rho,\, \I^{\star}\in \real,\, \rho>0.
\end{align}

This function has a long history in machine learning research and is a
commonly employed activation function in (artificial) neural networks
\citep{GBB:11}, especially for binary classification problems where the
output is interpreted~as~a~probability. The sigmoidal
activation function offers a smoothed transition compared to the sharp
threshold of the rectified hyperbolic tangent, providing more gradual
changes in neural firing rates.  As such, it makes sense to define the
\emph{activation current} $\I^{\star}$ as the key point where neural
firing rates start to considerably rise, as observable from
Fig.~\ref{fig:sigm}.  Mathematically, $\I^{\star}$ is defined by
$\phi(\I^{\star} + (2\rho)^{-1}) = 1/2$, indicating the $\I$-axis
intersection with the tangent to $\phi$ at its inflection point.  In this
case, the \emph{gain factor} $\rho$ is scaled by $4$ to match the maximal
slope of the sigmoid to that of the rectified hyperbolic tangent,
allowing a direct comparison of stability effects across both functions. 
The aim is to study how the stability of the system varies depending on
the values of the parameters $\rho$ and $\I^{\star}$ as they vary along
some fixed intervals.  As shown in Fig.~\ref{fig:PD1}(a,d), areas
satisfying the stability condition \eqref{eq:stability} align closely
with regions of observed numerical stability, suggesting the accuracy of
the analytical stability criterion.  Furthermore, the area for the
stability condition widens when $\gamma<0$, whereas it shrinks
considerably to a narrow region for $\gamma>0$. Figures \ref{fig:PD1}(b)
and \ref{fig:PD1}(e) plot the energy associated with \emph{firing rate}
models with, $\rho=4.8$, $\I^{\star}=0.2$ for stable retrievable memories
(blue square marker in Fig.~\ref{fig:PD1}(a)) and $\I^{\star}=0.8$ for
unstable retrievable memories (red starred marker in
Fig.~\ref{fig:PD1}(a)). Crucially, when network parameters guarantee the
stability of retrievable memories (Fig.~\ref{fig:PD1}(c)), the dynamics
of the \emph{firing rate} model naturally converge toward
them. Conversely, when the network parameters make the retrievable
memories saddle points for the dynamics, the system's activity is
repelled from these states and moves toward the origin
(Fig.~\ref{fig:PD1}(f)).  Finally, as shown by a comparison of
Fig.~\ref{fig:PD2}(a,d) and Fig.~\ref{fig:PD1}(a,d), the smoothness of
the sigmoidal activation function allows for broader stability regions
across the parameter space, offering enhanced stability compared to the
rectified hyperbolic tangent.

\begin{figure}[!tbph]
    \centering
    \subfloat[$\gamma<0$]{\includegraphics[width=.3\columnwidth]{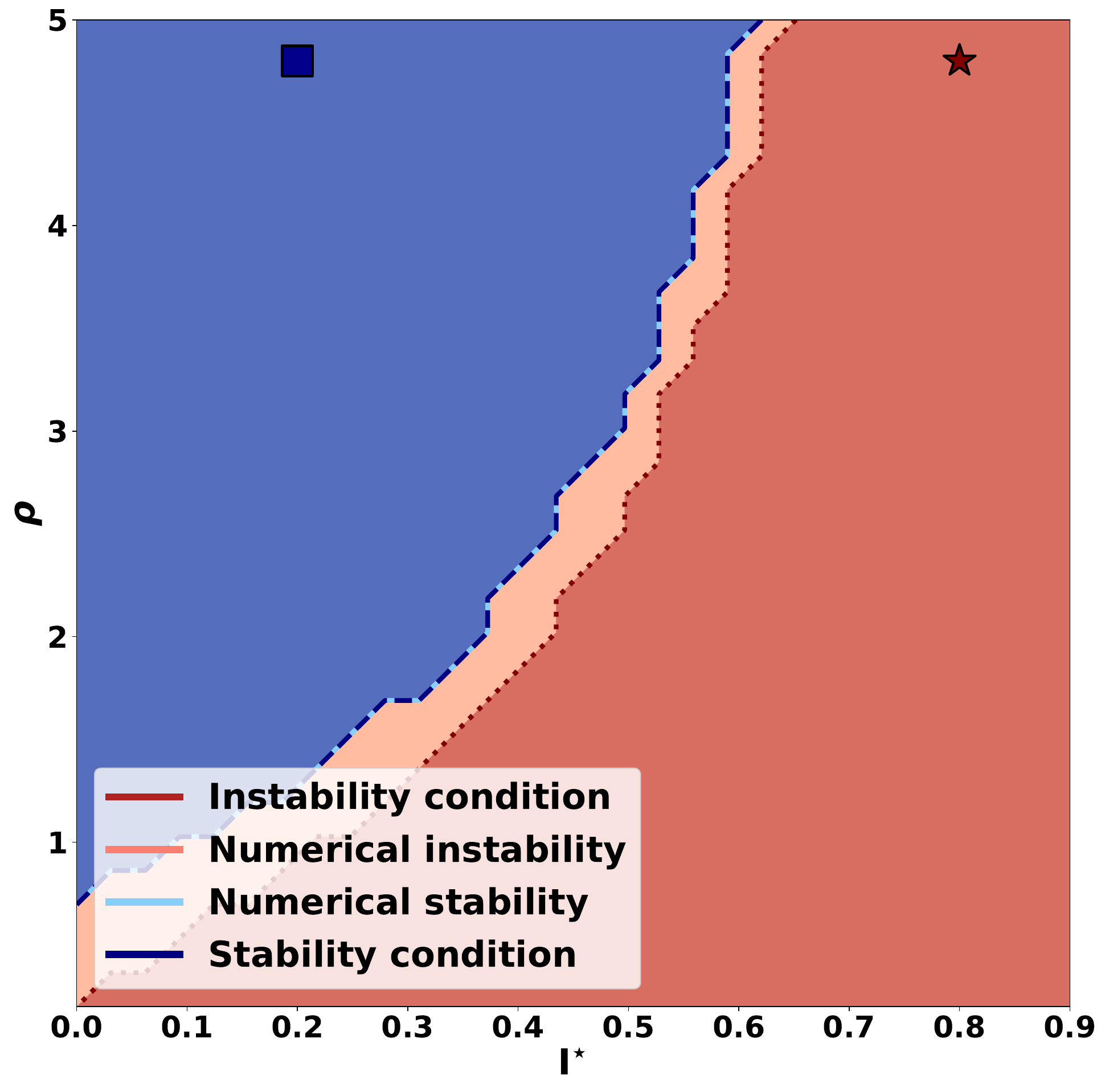}}
    \subfloat[Stable memories]{\includegraphics[width=.35\columnwidth]{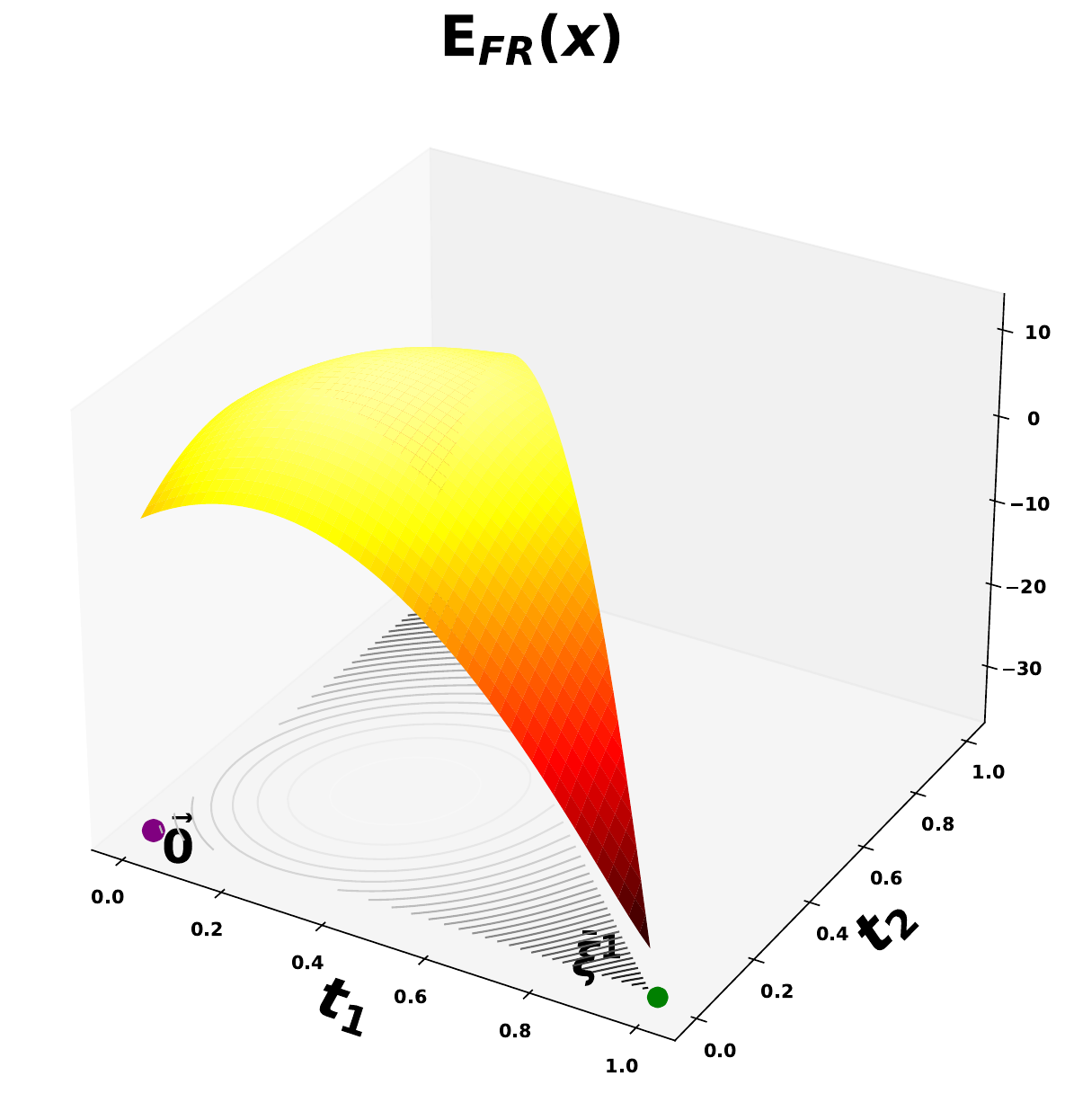}}
    \subfloat[Correct retrieval]{\includegraphics[width=.35\columnwidth]{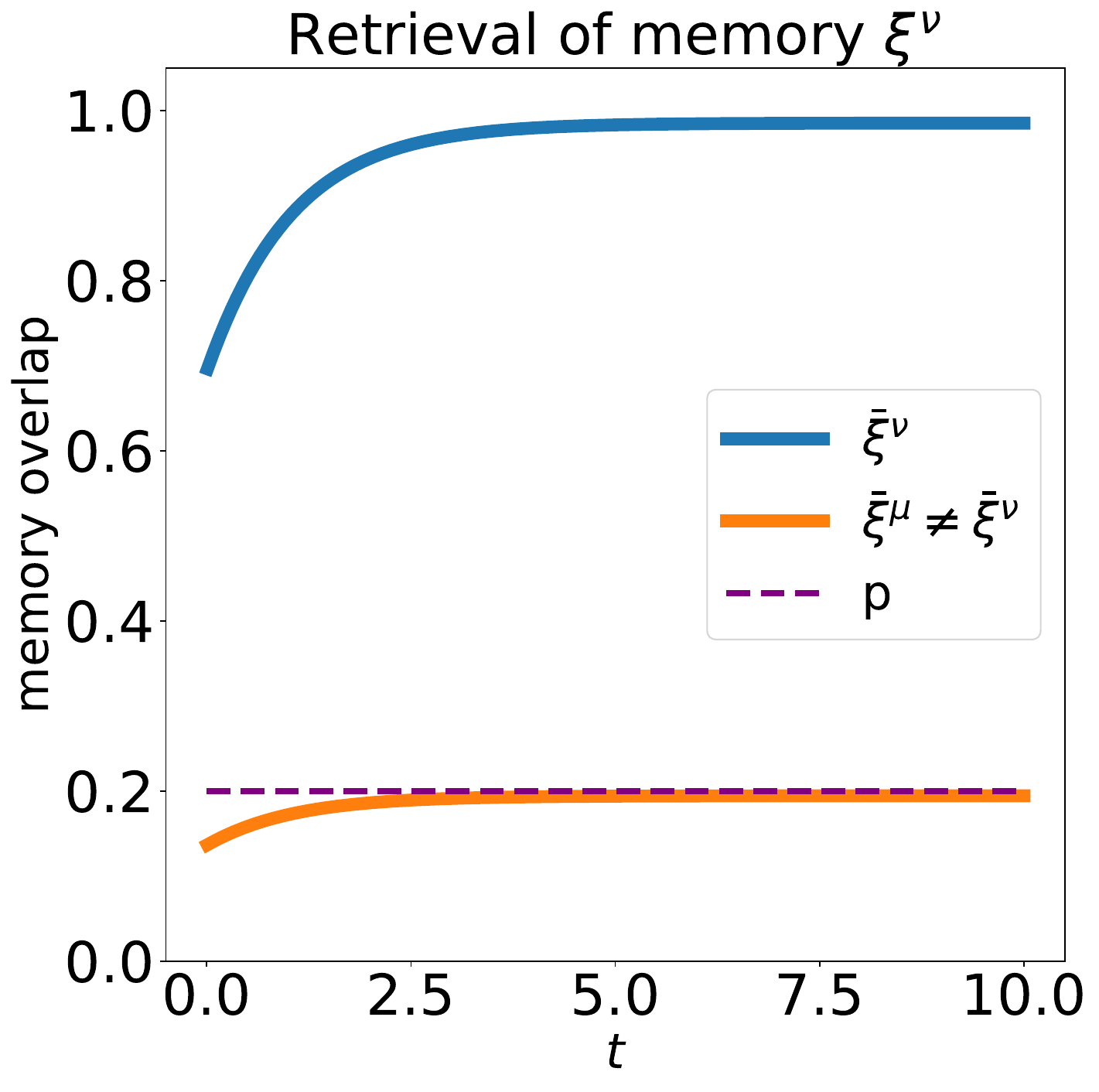}}
    
    \subfloat[$\gamma>0$]{\includegraphics[width=.3\columnwidth]{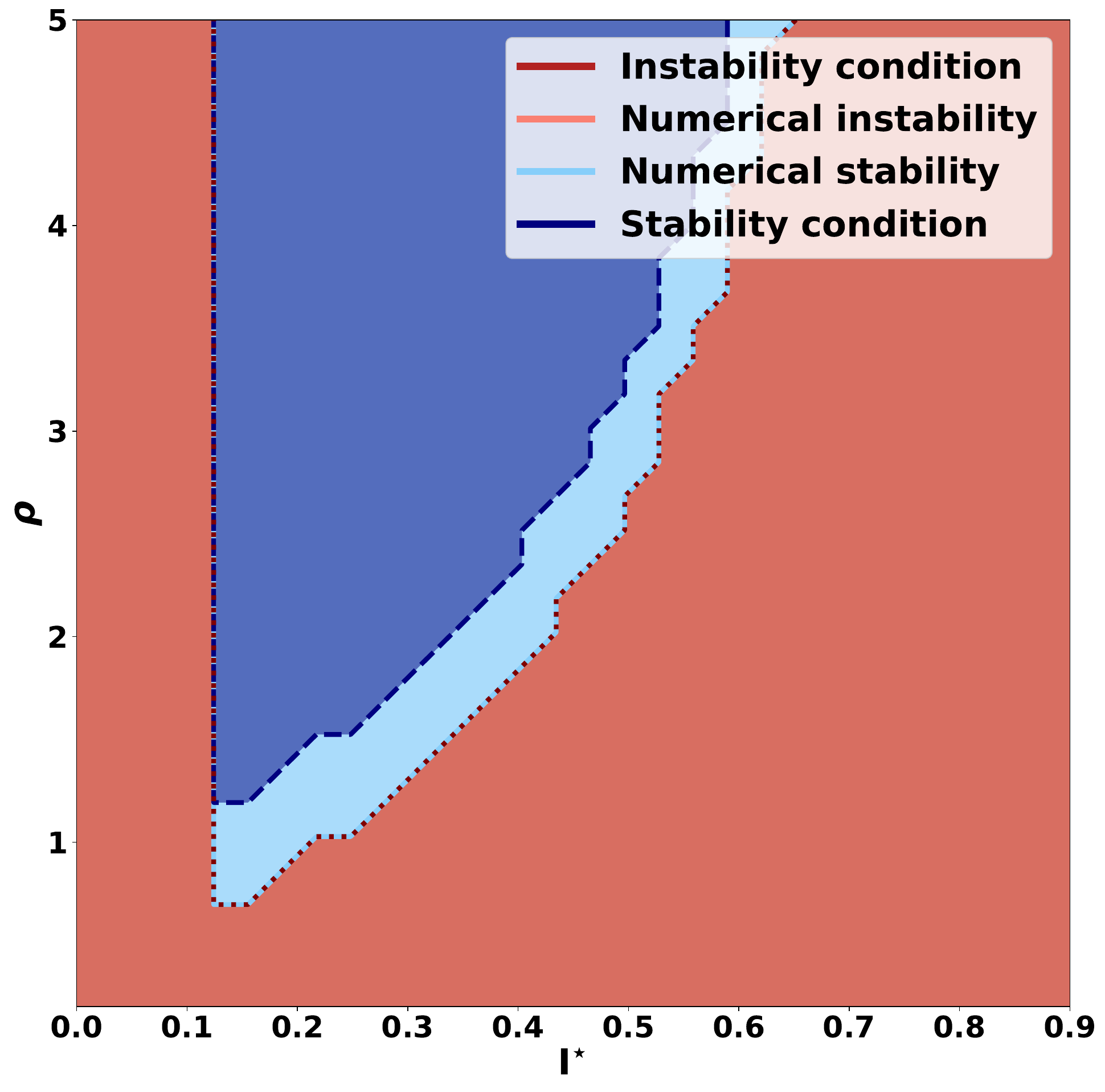}}
    \subfloat[Unstable memories]{\includegraphics[width=.35\columnwidth]{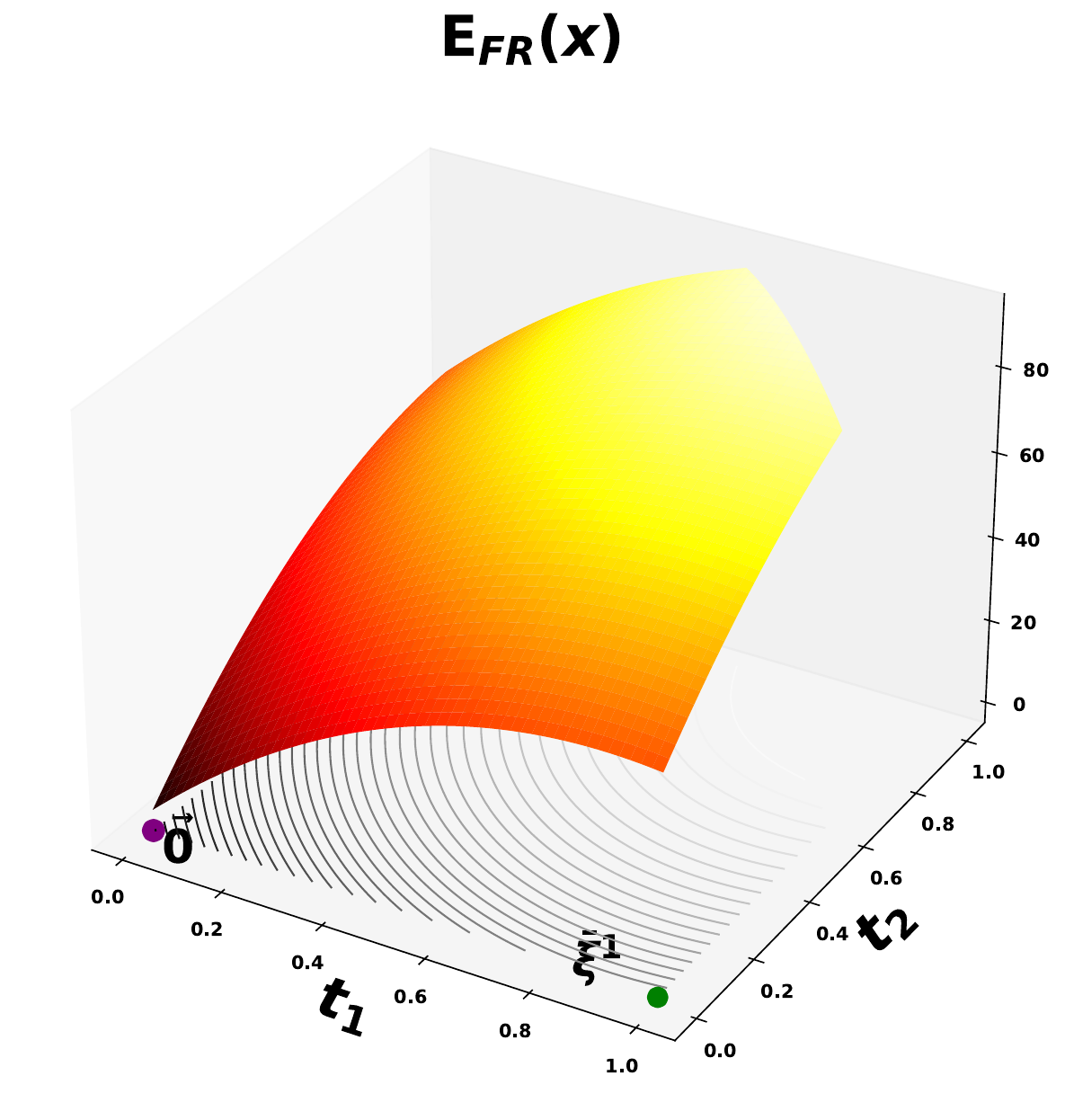}}
    \subfloat[Failed retrieval]{\includegraphics[width=.35\columnwidth]{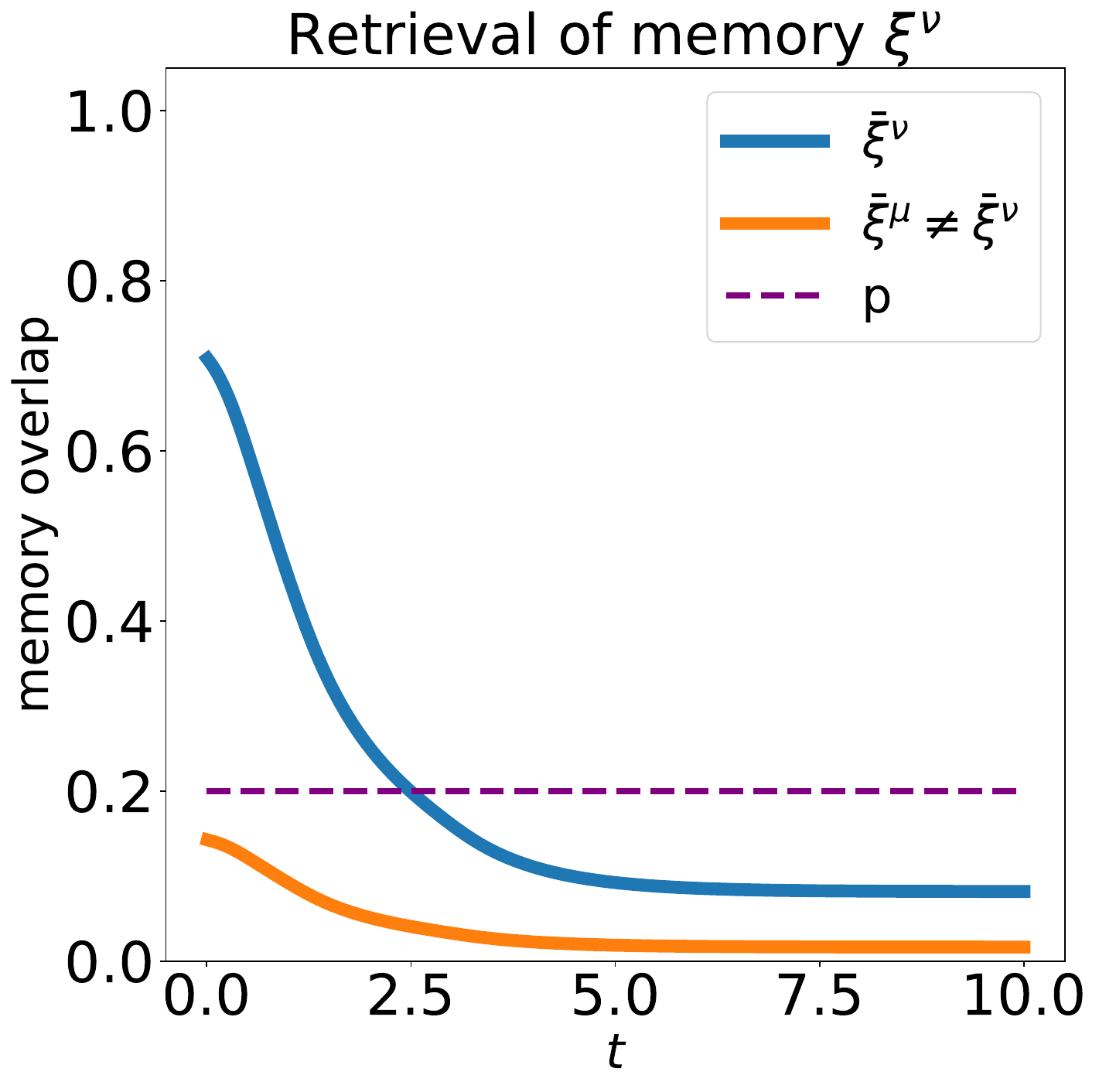}}

    \caption{Qualitative assessment of the design of \emph{firing rate}
      models with a rectified hyperbolic tangent activation function. (a, d) Phase diagram for the stability (or instability) of the \emph{firing rate} model for (a) $\gamma<0$ ($\I_{0}=-0.3$ and $\I_{1}=0.9$) and for (d) $\gamma>0$ ($\I_{0}=0.1$ and $\I_{1}=0.9$). The instability region is wide, and for $\gamma>0$, it consists of two disconnected region. Specifically, sufficiently small offsets lead to instability regardless of the
      slope. This is exactly the case in which $\gamma<\alpha$ and the point $(\I_{0},x_{0})$ is intersected from above, thus precluding stability. (b, e) Energy surface of the \emph{firing rate} model with $\gamma<0$ when the memories are (b) stable and (e) unstable. (c-f) Retrieval performances of the \emph{firing rate} model for $\gamma<0$ when the memories are (c) stable and (f) unstable. Parameters are set to $\rho=4.8$ and $\I^{\star}=0.2$ (stable memories) or $\I^{\star}=0.8$ (unstable memories).}
    \label{fig:PD2}
\end{figure}

\begin{figure}[tbph!]
    \centering
    \subfloat[$\gamma<0$]{\includegraphics[width=.3\columnwidth]{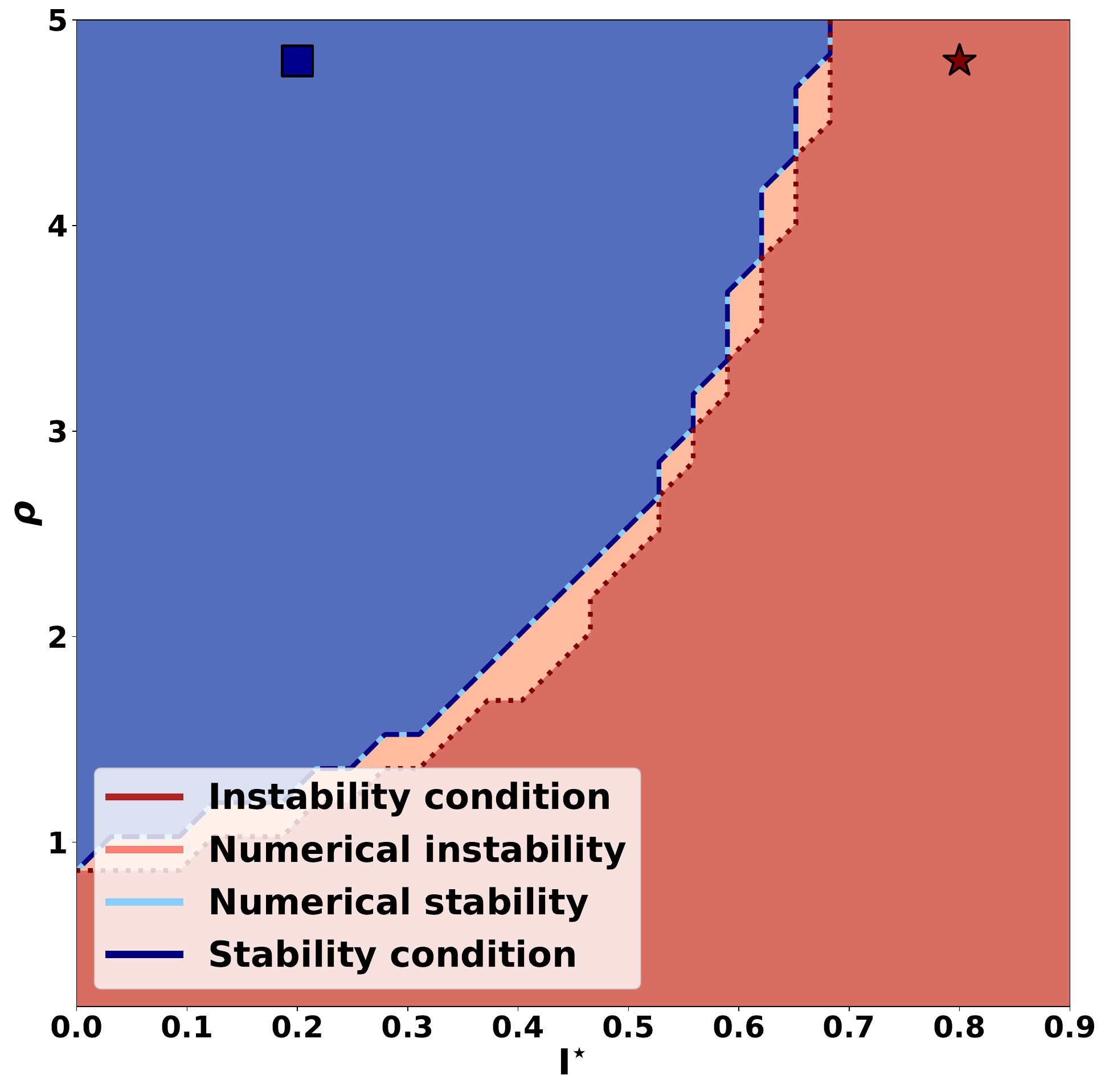}}
    \subfloat[Stable memories]{\includegraphics[width=.3\columnwidth]{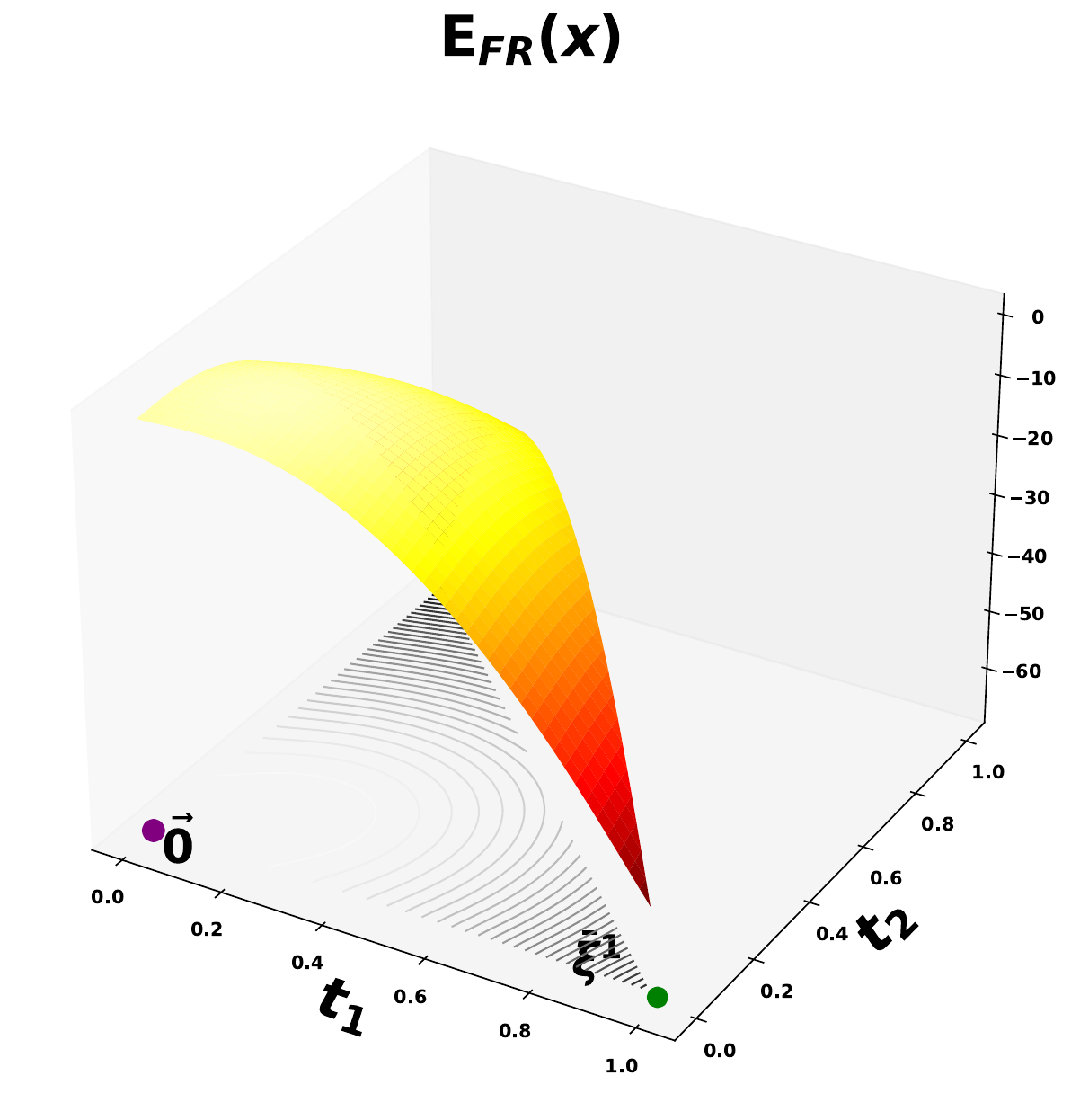}}
    \subfloat[Correct retrieval]{\includegraphics[width=.35\columnwidth]{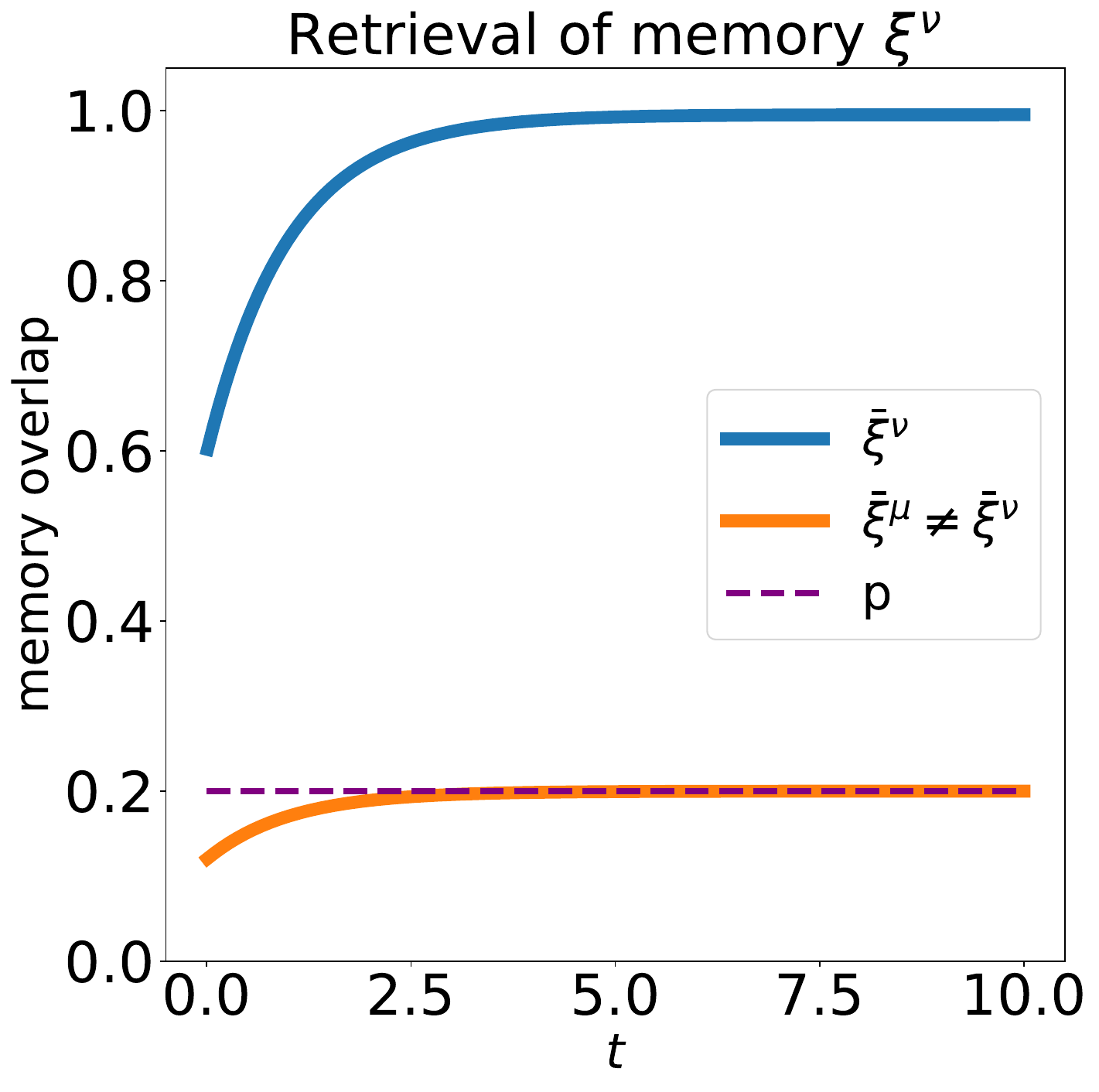}}
    
    \subfloat[$\gamma>0$]{\includegraphics[width=.3\columnwidth]{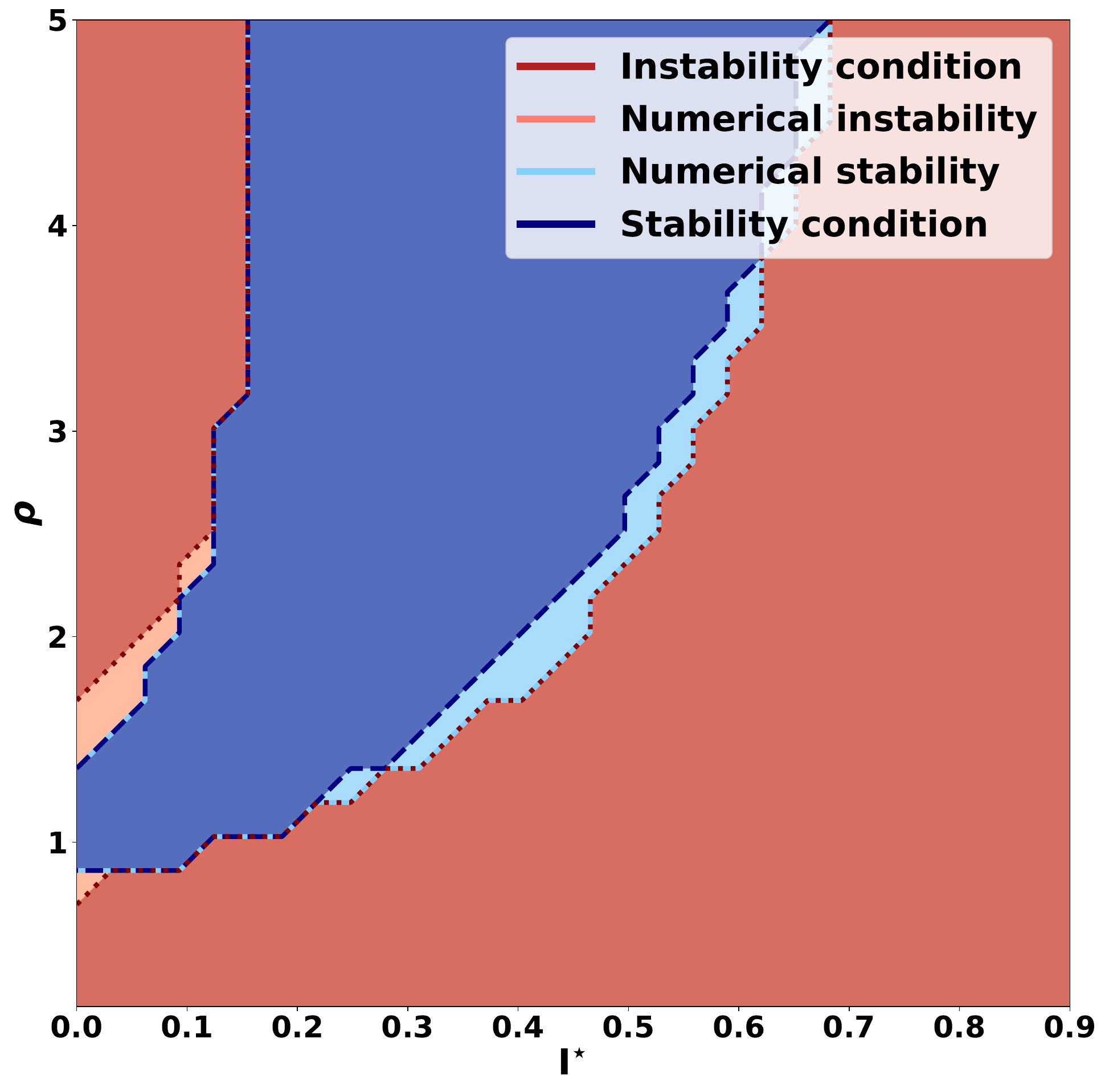}}
    \subfloat[Unstable memories]{\includegraphics[width=.3\columnwidth]{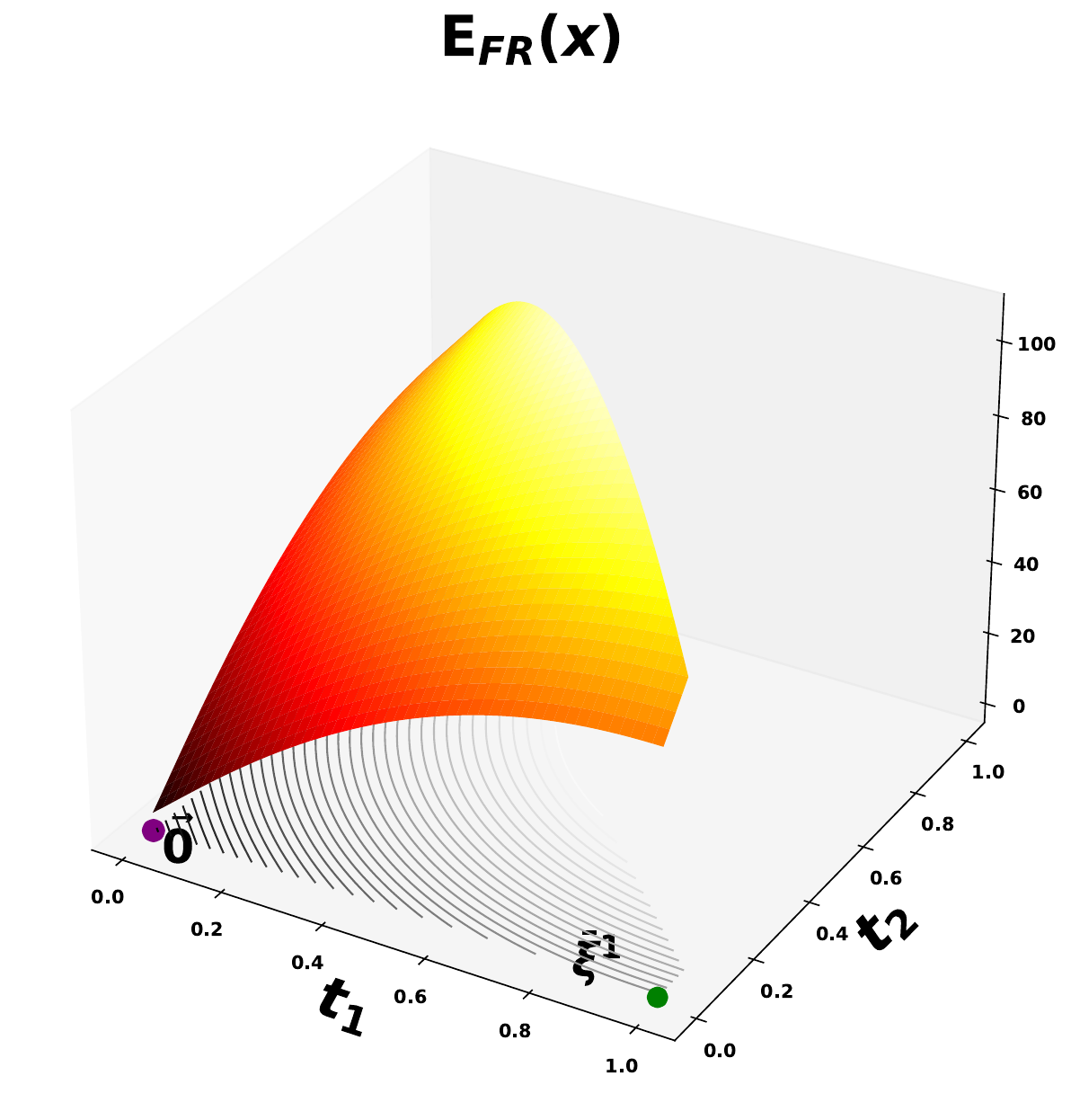}}
    \subfloat[Failed retrieval]{\includegraphics[width=.35\columnwidth]{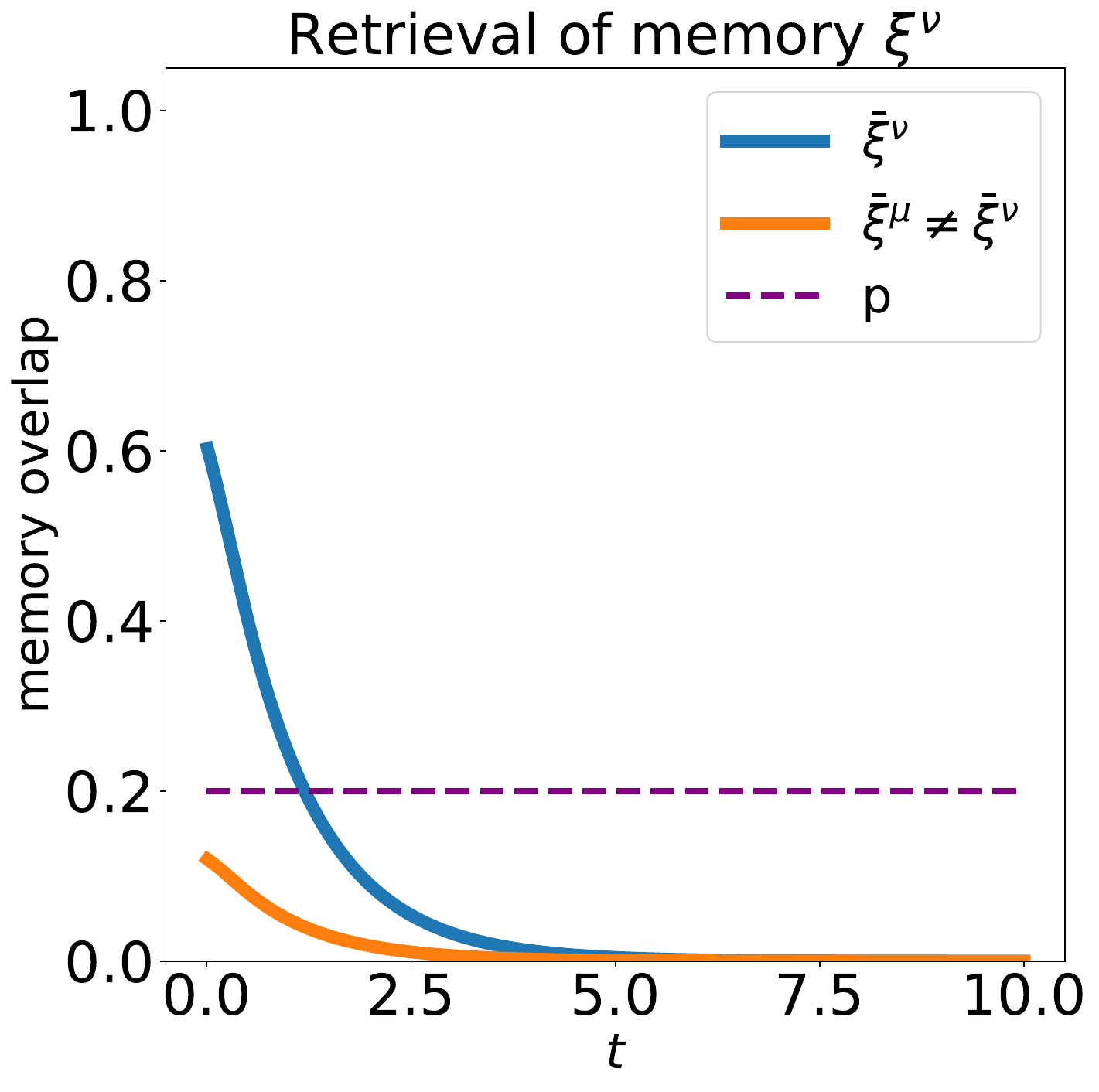}}

    \caption{Qualitative assessment of the design of \emph{firing rate} models with sigmoidal activation function. (a, d) Phase diagram for the stability (or instability) of the \emph{firing rate} model for (a) $\gamma<0$ ($\I_{0}=-0.3$ and $\I_{1}=0.9$) and for (d) $\gamma>0$ ($\I_{0}=0.1$ and $\I_{1}=0.9$). Specifically, for $\gamma>0$ we have two wide, disconnected instability regions, and a narrow stability region, highlighting how excitatory homeostatic strength significantly increase the model sensitivity to parameters choice. (c, f) Energy surface associated to \emph{firing rate} models with parameters lying in the (c) stability region of the phase portrait (f) instability region of the phase portrait. (c) For the choice of parameters $\rho=4.8$ and $\I^{\star}=0.2$, the memories of the \emph{firing rate} system are locally asymptotically stable and local minima for the Energy. (f) For the choice of parameters $\rho=4.8$ and $\I^{\star}=0.8$, the memories are unstable equilibria for the \emph{firing rate} model and saddle points of the associated Energy.}
    \label{fig:PD1}
\end{figure}

\section{Conclusion}

\emph{Firing rate} models provide a biologically
plausible and efficient framework for representing neural activity at the
population level, making them particularly well-suited for capturing key
macroscopic behaviors in neuronal networks, including associative memory
formation. This study introduces a positive \emph{firing rate} model that
encodes memories as neural firing rates, mirroring realistic activity
levels observed in the brain. Through a balanced excitatory-inhibitory
synaptic matrix, we introduced a covariance-based approach that allows
re-scaled memories to emerge as equilibrium points in network dynamics.
This approach enables a reinterpretation of well-known examples from the
literature \citep[Sec.~7.4]{AD:05} and resolves the ``anti-memory''
phenomenon—a drawback inherent to traditional Hopfield model designs—by
stabilizing only intended memories as equilibrium states.  Furthermore,
our analysis of local and global asymptotic stability conditions for
rescaled memories provides insight on the design choices that ensure
effectiveness and robustness of the \emph{firing rate} model as an
associative memory device.  This study applies rigorous mathematical
analysis typically reserved for \emph{voltage} models to \emph{firing
rate} frameworks, bridging a critical gap and enabling a unified
understanding of associative memory dynamics.  These results are easily
amenable to graphical representation, and therefore present both
practitioners and theoreticians with clear intuitive guidance on the
design choices that will result in an effective \emph{firing rate}
associative memory system.  This study employs a deterministic approach
to memory definition, diverging from traditional probabilistic
frameworks.  Future research could explore the model’s capacity in
probabilistic settings, specifically addressing the scalability and
storage limits of \emph{firing rate} associative memories.  In summary,
this work advances our understanding of the fundamental properties of
\emph{firing rate} networks and their application to associative memory,
providing both a theoretical contribution and practical insights for
researchers developing biologically plausible associative memory
networks.

\section*{Appendix}

\subsection*{On generalized prototypical memories and equilibria assignment}

We start from the proof of Theorem \ref{thm:W}.
We will prove this theorem in a more general setting than the one presented in Sections~3.
Indeed, we consider a set of prototypical memories that are $\{0,1\}$-valued vectors $\{\xi^{\mu}\}_{\mu=1}^P$  and we assume they satisfy the following conditions
\begin{subequations}\label{eq:memories-G}
\begin{align}
    &\text{(equal sparsity):}\qquad \vectorones[n]^\top \xi^{\mu}=p n , \ \ \ \mu=1,\dots,P, \label{eq:memories-b-G}\\
    &\text{(equal correlation):}\qquad {\xi^{\mu}}^\top \xi^{\nu} = p r n, \ \ \ \forall\, \label{eq:memories-c-G} \mu\ne \nu,
\end{align}
\end{subequations}
where $p\in (0,1)$ and $r\in [0,1)$.  Note that we recover the properties
  in equation~\eqref{eq:memories} by setting $r=p$.
\medskip

We now generalize the construction of the covariance-based synaptic matrix to the case of the new prototypical and associated retrievable memories. The new synaptic weight matrix is
\begin{align}\label{eq:W-genr}
    W = \!\frac{\alpha}{p(1-r) n}\! \sum_{\mu=1}^P (\xi^{\mu}\! -\! \beta \vectorones[n])(\xi^{\mu}\! -\! \beta \vectorones[n])^\top\!\!\! +\! \frac{\gamma}{n}\vectorones[n]\vectorones[n]^\top,
\end{align}
where $\alpha$, $\beta$, $\gamma\in \real$ are defined as
\begin{subequations}\label{eq:parameters-G}
    \begin{align}
    &\alpha:=\frac{\I_{1}-\I_{0}}{x_{1}-x_{0}},\label{eq:alpha-G}\\
    & \beta:=p\frac{r x_{1}+(1-r)x_{0}}{p x_{1}+(1-p)x_{0}},\label{eq:beta-G} \\
    &\gamma:= \frac{\beta \I_{1}+(1-\beta)\I_{0}}{px_{1}+(1-p)x_{0}}. \label{eq:gamma-G}
    \end{align}
\end{subequations}
Consequently, we are ready to generalize the theorem on the equilibria assignment as follows.
\begin{theorem}[Equilibria assignment through $W$]\label{thm:Wgen}
Consider a set of vectors $\{\xi^{\mu}\}_{\mu=1}^P$ satisfying \eqref{eq:memories}. Let $\I_{0},\I_{1}\in\real$, with $\I_{0}< \I_{1}$ and $x_{0}:=\phi(\I_{0})$ and $x_{1}:=\phi(\I_{1})$, and $W$ as in \eqref{eq:W-genr}, 
Then the vectors $\{\bar \xi^{\mu}\}_{\mu=1}^P$ defined in \eqref{eq:RM-components} are equilibria of \eqref{eq:FR}. \smallskip
\end{theorem}

\begin{proof}
Using the definitions of $\bar\xi^{\nu}$ and $W$ in \eqref{eq:W-map} and \eqref{eq:RM-components}, it is a matter of lengthy but straightforward calculations to check that the following identity holds
\begin{align*}
    W\bar\xi^{\nu} &=  \frac{\alpha}{p(1-r) n}\left(\xi^{\nu} - \beta \vectorones[n]\right)\left(\xi^{\nu} - \beta \vectorones[n]\right)^\top(x_{0} \vectorones[n] + (x_{1}-x_{0})\xi^{\nu}) \\[0.175cm] 
    & +\frac{\alpha}{p(1-r) n}\sum_{\mu\ne \nu}\left(\xi^{\mu} - \beta\vectorones[n]\right)\left(\xi^{\mu} - \beta\vectorones[n]\right)^\top(x_{0} \vectorones[n] + (x_{1}-x_{0})\xi^{\nu})\\
    &+\frac{\gamma}{ n} \vectorones[n]\vectorones[n]^\top (x_{0} \vectorones[n] + (x_{1}-x_{0})\xi^{\nu})\\[0.175cm]
    & = (\I_{1}-\I_{0})\xi^{\nu}+\I_{0}\vectorones[n].
\end{align*}
From the above identity, it follows that
\begin{align}
\Phi(W\bar\xi^{\nu}) &= \Phi((\I_{1}-\I_{0})\xi^{\nu}+\I_{0}\vectorones[n]) \notag\\
& = (x_{1}-x_{0})\xi^{\nu}+x_{0}\vectorones[n] = \bar \xi^{\nu},\label{eq:eq-condition}
\end{align}
where we used the fact that the entries of the vector $(\I_{1}-\I_{0})\xi^{\nu}+\I_{0}\vectorones[n]$ takes values in $\{\I_{0},\I_{1}\}$ and the identities $x_{0}=\phi(\I_{0})$, $x_{1}=\phi(\I_{1})$.
To conclude, observe that equation \eqref{eq:eq-condition} implies that the vectors $\{\bar \xi^{\mu}\}_{\mu=1}^P$ are equilibria of~\eqref{eq:FR}.
\end{proof}

\subsection*{On the linear independence of the memories}

From the geometric constraints placed on the prototypical memories, namely the equal sparsity and equal correlation constraints, the following proposition follows. 

\begin{proposition}[Linear Independence of the Memories]\label{prop: LI-mem}
Let $p,r\in(0,1)$ and  let the prototypical memories $\{\xi^{\mu}\}_{\mu=1}^{P}$ satisfy the equal sparsity (\ref{eq:memories-b-G}) and the equal correlation (\ref{eq:memories-c-G}) constraints. Then
\begin{itemize}
    \item[(i)] The prototypical memories $\{\xi^{\mu}\}_{\mu=1}^{P}$ are linearly independent.
    \item[(ii)] The retrievable memories $\{\bar\xi^{\mu}\}_{\mu=1}^{P}$ are linearly independent.
\end{itemize}    
\end{proposition}
\begin{proof}
  We start by proving the second statement.
  \begin{itemize}    
  \item[(ii)] Suppose that $x_{1}\neq{x_{0}}$ and that the retrievable memories $\bar\xi^{\mu}=(x_{1}-x_{0})\xi^{\mu}+x_{0}\vectorones[n]$ for $\mu=1,\dots,P$ are linearly dependent, so that there exists $\beta=(\beta_{1},\dots,\beta_{P})\in\real^{P}/\{\vectorzeros\}$ such that
        \begin{equation}
            \vectorzeros=\sum_{\mu=1}^{P}\beta_{\mu}\bar\xi^{\mu}
        \end{equation} 
        Without loss of generalization suppose that $\beta_{\nu}\neq{0}$, from which
        \begin{equation}
            \xi^{\nu}=-\sum_{\mu\neq\nu}\frac{\beta_{\mu}}{\beta_{\nu}}\xi^{\mu}-\frac{x_{0}}{x_{1}-x_{0}}\vectorones[n]\left(1+\sum_{\mu\neq\nu}\frac{\beta_{\mu}}{\beta_{\nu}}\right)
        \end{equation}
        Exploiting now the equal sparsity constraint (\ref{eq:memories-b-G}) we obtain that
        \begin{align}
            pn&=\vectorones[n]^{\top}\xi^{\nu}\nonumber\\
            &=-pn\sum_{\mu\neq\nu}\frac{\beta_{\mu}}{\beta_{\nu}}-pn\frac{x_{0}}{x_{1}-x_{0}}\frac{1}{p}\left(1+\sum_{\mu\neq\nu}\frac{\beta_{\mu}}{\beta_{\nu}}\right)
        \end{align}
        which implies that
        \begin{equation}
            \left(1+\sum_{\mu\neq \nu}\frac{\beta_{\mu}}{\beta_{\nu}}\right)pn=-pn\frac{x_{0}}{x_{1}-x_{0}}\frac{1}{p}\left(1+\sum_{\mu\neq\nu}\frac{\beta_{\mu}}{\beta_{\nu}}\right)
        \end{equation}
        If $\sum_{\mu}^{P}\beta_{\mu}\neq0$, then simplifying we derive the following contradiction
        \begin{equation}
            1=-\frac{1}{p}\frac{x_{0}}{x_{1}-x_{0}}
        \end{equation}
        since by hypothesis $p>0$ and $x_{0}/(x_{1}-x_{0})\geq{0}$. Then we are only left to verify that the hypothesis $\sum_{\mu}^{P}\beta_{\mu}=0$ leads to another contradiction. Using again the hypothesis $\beta_{\nu}\neq 0$ and expressing $\beta_{\nu}=-\sum_{\mu\neq \nu}\beta_{\mu}$, we write the condition of linear dependence as
        \begin{align}
            \vectorzeros&=\sum_{\mu\neq \nu}\beta_{\mu}(\bar \xi^{\mu}-\bar \xi^{\nu})\nonumber\\
                        &=\sum_{\mu\neq \nu}\beta_{\mu}(x_{1}-x_{0})(\xi^{\mu}-\xi^{\nu})
        \end{align}
        from which it is immediate to reach the contradiction
        \begin{align}
            0 &= \vectorzeros^{\top}\xi^{\nu}\nonumber\\
              &= (x_{1}-x_{0})\sum_{\mu\neq \nu}\beta_{\mu}(\xi^{\mu}-\xi^{\nu})^{\top}\xi^{\nu}\nonumber\\
              &= pn\underbrace{(x_{1}-x_{0})}_{\neq{0}}\underbrace{(1-r)}_{\neq{0}}\underbrace{\left(-\sum\nolimits_{\mu\neq \nu}\beta_{\mu}\right)}_{\beta_{\nu}}
        \end{align}
        \item[(i)] To see that also the prototypical memories are linearly independent, it suffices to take $x_{1}=1$ and $x_{0}=0$. Then this reduces to verify that $\sum_{\mu}^{P}\beta_{\mu}=0$, which we have already proven leads to a contradiction.
        \end{itemize}
\end{proof}

\subsection*{When the antimemories are equilibria} 

We show now that antimemories are equilibria if and only if condition \eqref{eq:condantim} holds true. Here for simplicity we assume that $r=p$. First notice that $\bar\xi+\bar\xi^\text{ant}=(x_{0}+x_{1})\vectorones[n]$ so that $\bar\xi^\text{ant}=(x_{0}+x_{1})\vectorones[n]-\bar\xi$ and hence
$W\bar\xi^\text{ant}=(x_{0}+x_{1})W\vectorones[n]-W\bar\xi$. Since we know that $W\vectorones[n]=\gamma\vectorones[n]$ and $W\bar\xi=(\I_{1}-\I_{0})\xi+\I_{0}\vectorones[n]$, we can argue that $W\bar\xi^\text{ant}=((x_{0}+x_{1})\gamma-\I_{0})\vectorones[n]-(\I_{1}-\I_{0})\xi$. 

Now if $i$ is such that $\xi_i=0$ and $\xi^\text{ant}_i=1$, then, after some computations we see that
$$(W\bar\xi^\text{ant})_i=\frac{p\I_{1}x_{1}+p\I_{1}x_{0}+(1-2p)\I_{0}x_{1}}{px_{1}+(1-p)x_{0}}$$
In order $\bar\xi^\text{ant}$ to be an equilibrium point we need to impose that $\phi((W\bar\xi^\text{ant})_i)=\bar\xi^\text{ant}_i=x_{1}$ that holds if $(W\bar\xi^\text{ant})_i=\I_{1}$ which is equivalent to
$$\frac{p\I_{1}x_{1}+p\I_{1}x_{0}+(1-2p)\I_{0}x_{1}}{px_{1}+(1-p)x_{0}}=\I_{1}$$
After some simple computations we see that this is equivalent to the condition
\begin{equation}\label{eq:antimemcond}
  (1-2p)(\I_{0}x_{1}-\I_{1}x_{0})=0  
\end{equation}
Notice that this holds if and only if $p=1/2$ or $\I_{0}x_{1}=\I_{1}x_{0}$. If we start from the hypothesis that $i$ is such that $\xi_i=1$ and $\xi^\text{ant}_i=0$, then, after some computations we see that the equilibrium point condition holds again if and only if \eqref{eq:antimemcond} holds true. We can argue that this condition is necessary and sufficient for having that all the antimemories provide equilibrium points for \eqref{eq:FR}. Notice that, while this condition holds true for the Hopfield models, it generically fails for general \emph{firing rate} models when $p\not=1/2$. 

\subsection*{Memories not used in the construction of $W$}
Given a set of prototypical memories $\{\xi^{\mu}\}_{\mu=1}^{P}$, consider a vector $\xi\in\real^n$ different from each $\xi^{\mu}$ and with the property that the augmented set $\{\xi, \xi^{\mu}\}_{\mu=1}^{P}$ satisfies the equal sparsity and equal correlation Assumption~\ref{assum:memories}.  Define the covariance-based synaptic matrix $W$ as a function of the prototypical memories $\{\xi^{\mu}\}_{\mu=1}^{P}$ and the rescaled vector $\bar\xi=(x_{1}-x_{0})\xi+x_{0}\vectorones[n]$. It is straightforward to verify that $W\bar\xi = [p{\I_{1}}+(1-p)\I_{0}]\vectorones[n]$ from which it follows that $\Phi(W\bar\xi)=\phi(p{\I_{1}}+(1-p)\I_{0})\vectorones[n]\neq{\bar\xi}$. Hence, spurious equilibria such as $\xi$ cannot exist.

\subsection*{Existence of homogeneous equilibria}

We prove now Lemma \ref{lm:hom} on the existence of homogeneous equilibria for the system \ref{eq:FR}. We assume that $r=p$.

\begin{proof}[Proof of Lemma \ref{lm:hom}]
(i) Let $f(z):=\phi(z)-\gamma^{-1} z$ that is continuous. Using the definition of $\gamma$ in case $r=p$ it can be seen that
$$f(\I_{0})f(\I_{1})=-p(1-p)\frac{(\I_{0}x_{1}-\I_{1}x_{0})^2}{(p \I_{1}+(1-p)\I_{0})^2}\le 0$$
Then $f(\I_{0}),f(\I_{1})$ have opposite signs and by continuity of $f$ there must be $z\in[\I_{0},\I_{1}]$ such that $f(z)=0$ and hence such $z$ satisfies equation \eqref{eq:puntfiss}.

(ii)   Assume now that $z$ satisfies equation \eqref{eq:puntfiss} and let $\bar x:=\gamma^{-1}z\vectorones[n]$. Since when $r=p$ we have that $W\vectorones[n]=\gamma \vectorones[n]$, then $W\bar x=z\vectorones[n]$. Hence $\Phi(W\bar x)=\Phi(z\vectorones[n])=\phi(z)\vectorones[n]=\gamma^{-1}z\vectorones[n]=\bar x$ proving in this way that $\bar x$ is an equilibrium point.
\end{proof}

\subsection*{The local stability result}

We prove here Theorem \ref{thm:stab} in the general case in which $r\not= p$.

\begin{theorem}[Local stability condition for generalized prototypical memories]\label{thm:stab-gen}
    Consider a set of vectors $\{\xi^{\mu}\}_{\mu=1}^P$ satisfying \eqref{eq:memories} and let $W$ be as in Theorem \ref{thm:Wgen}. Assume that $\phi(\cdot)$ in \eqref{eq:FR} satisfies Assumption 1 and let 
    \begin{align}
        &b:=\frac{\I_{0}x_{1}-\I_{1}x_{0}}{x_{1}-x_{0}}\\
        &c:=\alpha\frac{p-r}{1-p}\frac{1}{px_{1}+(1-p)x_{0}}\\
        &d:=\alpha\frac{p-r}{1-p}\frac{\beta x_{1}+(1-\beta)x_{0}}{px_{1}+(1-p)x_{0}}\\
        &\eta:=\max\{\phi'(\I_{0}),\phi'(\I_{1})\}
    \end{align}
    Let $z_1,z_2$ be the roots of the polynomial 
    $$\mathcal{P}(z):=(z-\gamma)(z-\alpha)+P(dz-\alpha d-bc)$$
 \begin{itemize}
\item[(i)]   If $z_1,z_2$ are real then
    \begin{equation}\label{eq:stability-new}
    \eta\max\{\alpha,z_1,z_2\}<1
     \end{equation}
    implies that the equilibria $\{\bar{\xi}^{\mu}\}_{\mu=1}^P$ of \eqref{eq:FR} are locally asymptotically stable.
\item[(ii)]    If $z_1,z_2$ are are not real then
    \begin{equation}\label{eq:stability-newalt}
    \eta\alpha<1
     \end{equation}
    implies that the equilibria $\{\bar{\xi}^{\mu}\}_{\mu=1}^P$ of \eqref{eq:FR} are locally asymptotically stable.
 \end{itemize}   
\end{theorem}

Note that, if $r=p$, then $c=d=0$, so that $z_1=\gamma$, $z_2=\alpha$ and condition \eqref{eq:stability-new} reduces to \eqref{eq:stability}. Note moreover that case (ii) in the previous theorem is nongeneric in the sense that it applies in a very specific situation. \footnote{ Namely, the equality $Pp-(P-1)r=1$ has to hold.}.

Before proving the previous theorem we propose a lemma.

\begin{lemma}\label{lemma:stab-gen}
    Under the same assumptions and using the same notation of Theorem \ref{thm:stab-gen} we have that, if
    \begin{equation}\label{eq:stabdis}
    \eta W \prec I_n
    \end{equation}
        then the equilibria $\{\bar{\xi}^{\mu}\}_{\mu=1}^P$ of \eqref{eq:FR} are locally asymptotically stable.
 \end{lemma}

\begin{proof}
It is well known that, according to Lyapunov's indirect method, a sufficient condition that ensures local asymptotic stability of an
equilibrium point $\bar{\xi}^{\mu}$ is that the Jacobian matrix \eqref{eq:Jacobian} evaluated at $\bar{\xi}^{\mu}$ has eigenvalues with negative real part. Observe that the Jacobian matrix can be written as
\begin{align}
    J(\bar{\xi}^{\mu}) &= -I_{n}+\diag(\Phi'(W\bar{\xi}^{\mu}))W \notag\\
    &= -I_{n}+\underbrace{\diag(\Phi'((\I_{1}-\I_{0})\xi^{\mu}+\I_{0}\vectorones[n]))}_{=:D}W \label{eq:J-stab}
\end{align}
where we used the identity $W\bar{\xi}^{\mu}=(\I_{1}-\I_{0})\xi^{\mu}+\I_{0}\vectorones[n]$ established in the proof of Theorem \ref{thm:Wgen}. 
This means that $D$ is a diagonal matrix with diagonal entries that are $\phi'(\I_{0})$ or $\phi'(\I_{1})$. From Assumption 1 we know that $\eta\ge 0$. 
We distinguish three cases:\\ 
a) If $\eta=0$, then \eqref{eq:stabdis} holds true. Moreover, $D=0$ and hence $J(\bar{\xi}^{\mu})=-I_{n}$ which implies that $\bar \xi^{\mu}$ is locally asymptotically stable, proving that the thesis holds true in this case. \\
b) Assume now that $\eta>0$ but that one of the two derivatives $\phi'(\I_{0})$ or $\phi'(\I_{1})$ is zero. This means that some of the diagonal entries of $D$ are zero. Assume with no loss of generality that the first $n_1$ diagonal entries of $D$ are positive. Hence we can write 
$$D= \begin{bmatrix}
    D_1 &  0\\
    0 & 0
\end{bmatrix}. $$
where $D_1$ is diagonal and with positive diagonal entries. Then
\begin{align*}
    J(\bar{\xi}^{\mu}) &= -I_{n}+DW =
    \begin{bmatrix}
    -I_{n_{1}} &  0\\
    0 & -I_{n_{2}}
    \end{bmatrix}
    +\begin{bmatrix}    D_1 &  0\\
    0 & 0
    \end{bmatrix}
    \begin{bmatrix}
    W_{11} &  W_{12} \\
    W_{21}  & W_{22} 
    \end{bmatrix}\\
&=
\begin{bmatrix}
    -I_{n_{1}}+D_1W_{11} &  D_1W_{12} \\
    0  & -I_{n_2}
\end{bmatrix}.
\end{align*}
where $n_{1},n_{2}\in\mathbb{N}$ and $n_{1}+n_{2}=n$.
Hence the eigenvalues of
$J(\bar{\xi}^{\mu})$ have negative real part if and only if
$-I_{n_1}+D_1W_{11}$ has this property and this happens if and only if
$D_1^{1/2}W_{11}D_1^{1/2}\prec I_{n_{1}}$ or, equivalently, if and only if
$W_{11}\prec D_1^{-1}$. In this way we have shown that $W_{11}\prec D_1^{-1}$ implies the local asymptotic stability of the
equilibrium point $\bar{\xi}^{\mu}$. Now observe that condition \eqref{eq:stabdis} implies $W_{11}\prec \eta^{-1}I_{n_1}$. Since $ D_1\preceq \eta I_{n_1}$, then $ D_1^{-1}\succeq \eta^{-1}I_{n_1}$ and hence we can argue that $W_{11}\prec D_1^{-1}$ that is what we need for proving the local asymptotic stability.\\
c) Assume now that both the derivatives $\phi'(\I_{0})$ and $\phi'(\I_{1})$ are both nonzero. This means that the matrix $D$ is invertible. The proof of the local asymptotic stability can be obtained following the same arguments used in the previous point. 
\end{proof}

\begin{proof}[Proof of Theorem \ref{thm:stab-gen}] Let $\lambda_\text{max}(\cdot)$ denote the maximum eigenvalue of a symmetric matrix. According the previous lemma, $\lambda_\text{max}(\eta W)=\eta\lambda_\text{max}(W)<1$ implies local asymptotic stability.
Hence the
theorem is proved if we prove that the eigenvalues of $W$ may only take one
of the values $0,\alpha,z_1,z_2$.
To this aim consider the matrix $T:=[\vectorones[n],\bar \xi^1,\ldots,\bar\xi^P]\in\real^{n\times (P+1)}$. 
Since it can be verified that $W\vectorones[n]=c\sum_{\mu=1}^P\bar\xi^\mu+(\gamma-Pd)\vectorones[n]$,
$W\bar\xi^\nu=\alpha\bar\xi^\nu+b\vectorones[n]$  for all $\nu=1,\ldots,P$, and $Ww=0$ for all $w\in \mathrm{Im}(T)^\perp$,
then we can argue that the subspaces $\mathrm{Im}(T)$ and $\mathrm{Im}(T)^\perp$ are $W$-invariant.
Moreover, it holds
$$WT=T\underbrace{\begin{bmatrix}
    \gamma-Pd & b\vectorones[P]^T \\
    c\vectorones[P] & \alpha I_{P} 
  \end{bmatrix}}_{=:~\bar W\in\real^{(P+1)\times (P+1)}}$$
Let now $y\neq{\vectorzeros}$ such that $Wy=\lambda y$ for $\lambda\neq 0$ and observe that $y\in \text{Im}(T)$, since $\text{Im}(W)\subseteq\text{Im}(T)$. Define $x:=T^{\top}y$ and observe that, since $y\in \text{Im}(T)$, then, $x\neq \vectorzeros$. 
We have that
\begin{align}
    \lambda x^{\top} &= \lambda y^{\top} T\nonumber\\
    &= y^{\top}WT = y^{\top}T\bar W = x^{\top}\bar W
\end{align}
We proved in this way that $\text{eig}(W)/\{0\}\subseteq\text{eig}(\bar W)/\{0\}$, where $\text{eig}(\cdot)$ means the set of eigenvalues of a square matrix. The characteristic polynomial of $\bar W$ is $(z-\alpha)^{P-1}((z-\gamma +Pd)(z-\alpha)-Pbc)$ that has roots $\alpha,z_1,z_2$. The thesis follows from the fact that $\alpha\ge 0$.
\end{proof}

\subsection*{The global stability result}

The following result is instrumental for the proof of Theorem \ref{thm:conv}.

\begin{lemma}\label{prop:invariance}
Let Assumption \ref{assum:activation-function} hold and assume that $\phi(\cdot)$ takes values in a bounded interval $\mathcal{I}$. 
Then there exists a large enough $M$ such that the set $[\phi(-M),\phi(M)]^n$ is forward invariant for the dynamics \eqref{eq:FR}.
\end{lemma}

\begin{proof}
Since $\mathcal{I}$ is bounded, then there exists $x_M>0$ such that $\mathcal{I}\subseteq[-x_M,x_M]^n$. Let $\bar M:=\|W\|_\infty x_M$, where $\|\cdot\|_\infty$ is the infinity norm of a vector or of a matrix. We want to prove that for any $M\ge \bar M$ we have that $[\phi(-M),\phi(M)]^n$ is forward invariant for the dynamics \eqref{eq:FR}. First observe that if $x\in [\phi(-M),\phi(M)]^n$, then $x\in \mathcal{I}^n\subseteq[-x_M,x_M]^n$ and hence $\|Wx\|_\infty\le \|W\|_\infty \|x\|_\infty=\bar M$ which implies that $\phi([Wx])\in [\phi(-\bar M),\phi(\bar M)]^n\subseteq [\phi(-M),\phi(M)]^n$.
Let now $x\in [\phi(-M),\phi(M)]^n$ such that $x_i=\phi(-M)$. Then
$$\dot{x}_i=-x_i+\phi([Wx]_{i})\ge-\phi(-M)+\phi(-M)=0$$
On the other hand, if we take instead $x\in [\phi(-M),\phi(M)]^n$ such that $x_i=\phi(M)$, then
$$\dot{x}_i=-x_i+\phi([Wx]_{i})\le-\phi(M)+\phi(M)=0$$
The above conditions imply that for any $x$ belonging to the boundary of $[\phi(-M),\phi(M)]^n$ the direction of the vector derivative in the dynamics \eqref{eq:FR} is either tangent or points inside the set $[\phi(-M),\phi(M)]^n$. The forward invariance of $[\phi(-M),\phi(M)]^n$ then follows from Nagumo's Theorem \citep[Theorem 4.7]{BM:08}.
\end{proof}

We next present the proof of Theorem \ref{thm:conv}.

\begin{proof}[Proof of Theorem \ref{thm:conv}]
By Lemma \ref{prop:invariance}, for any $x(0)\in\mathcal{I}^n$, there exists $M>0$ such that the set $\mathcal{S}:=[\phi(-M),\phi(M)]^n$ is forward invariant for \eqref{eq:FR} and contains $x(0)$. This implies that $\subscr{E}{FR}(x)$ in \eqref{eq:energy-FR} and $\nabla \subscr{E}{FR}(x)=(-Wx + \Phi^{-1}(x))^{\top}$ are well-defined when evaluated along the trajectories of \eqref{eq:FR}.

The derivative of $\subscr{E}{FR}$ along each such trajectory is
\begin{align*}
    \subscr{\dot{E}}{FR}(x) &= \nabla \subscr{E}{FR}(x) \dot{x}\\
    & = (-Wx + \Phi^{-1}(x))^\top (-x + \Phi(Wx))\\
    & = \sum_{i=1}^{n} (-[Wx]_{i}+\phi^{-1}(x_{i}))(-x_{i}+\phi([Wx]_{i})),
\end{align*}
where we used the symmetry of $W$. 
Since $\phi(\cdot)$ is  weakly increasing and $\phi^{-1}(\cdot)$ is a right inverse of $\phi(\cdot)$, if $\phi^{-1}(x_{i})\ge [Wx]_{i}$ then $\phi(\phi^{-1}(x_{i}))= x_{i}\geq \phi([Wx]_{i})$, and if $\phi^{-1}(x_{i})\le [Wx]_{i}$ then  $x_{i}\le[Wx]_{i}$. This implies that $(-[Wx]_{i}+\phi^{-1}(x_{i}))(-x_{i}+\phi([Wx]_{i}))\le 0$ for all $i=1,\dots,n$ and for all $x\in\mathcal{S}$. As a consequence, $\subscr{\dot{E}}{FR}(x)\le 0$ for all $x\in\mathcal{S}$. Moreover, $\subscr{\dot{E}}{FR}(x)=0$ if and only if $\phi^{-1}(x_{i})=[Wx]_{i}$ or $\phi([Wx]_{i})=x_i$ for all $i=1,\dots,n$. Since $\phi^{-1}(x_{i})=[Wx]_{i}$ implies $\phi([Wx]_{i})=x_i$, we conclude that $\subscr{\dot{E}}{FR}(x)=0$ if and only if $\phi([Wx]_{i})=x_i$ for all $i=1,\dots,n$, that is, if and only if $x$ is an equilibrium point of \eqref{eq:FR}. Since $\subscr{E}{FR}(x)$ is Lipschitz in $\mathcal{S}$, the thesis follows from the LaSalle invariance principle  
\citep[Corollary 1]{L:68}. 
\end{proof}

\subsection*{The construction of the deterministic memories}
We now provide a simple yet effective way to build a set of memories that satisfy the equal sparsity \eqref{eq:memories-b} and equal correlation \eqref{eq:memories-c} constraints. Fix $n\in\mathbb{N}$ the size of the network and $P\in\mathbb{N}$ the number of memories to be stored. Then we set the equal sparsity parameter as $p=1/(P-1)$ and define the matrix $M\in\real^{n\times P}$ having the prototypical memories as column as
$$
    M=
    \begin{bmatrix}
        \vectorones[p^{2}n] 
        \vectorones[P]^{\top} \\
        I_{P} \\
        \vdots \\
        I_{P}
    \end{bmatrix}
$$
with the identity matrices being repeated $p(1-p)n$ times. Then the columns of $M$ provide a set of $P$ memories satisfying the equal sparsity \eqref{eq:memories-b} and equal correlation \eqref{eq:memories-c} constraints.

\bibliographystyle{apa}

\end{document}